\documentclass[12pt]{article}

\usepackage{jheppub}
\makeatletter
\def\@fpheader{\relax}
\makeatother

\usepackage{amsthm}
\usepackage{amsthm,amsmath,amssymb,mathtools}
\usepackage{mathrsfs}
\usepackage{bm}
\usepackage{bbm}
\usepackage{graphicx}
\usepackage{makecell}
\usepackage{float}
\usepackage{caption}
\usepackage{subcaption}
\usepackage[most]{tcolorbox}
\usepackage[mathscr]{euscript}
\usepackage{dsfont}
\usepackage[mathscr]{euscript}
\usepackage{mathrsfs}
\usepackage{hyperref}
\usepackage{tikz}
\newtheorem{theorem}{Theorem}[section]

\newtheorem{proposition}[theorem]{Proposition}  

\theoremstyle{definition}

\theoremstyle{remark}

\def\bea{\begin{eqnarray}}
\def\eea{\end{eqnarray}}
\def\cZ{\mathcal{Z}}
\def\cO{\mathcal{O}}
\def\cP{\mathcal{P}}
\def\cN{\mathcal{N}}
\def\cS{\mathcal{S}}
\def\Coeff{ {\rm Coeff } }
\def\simexp{ \underset{\mathrm{(e.t.)}}{\sim} } 

\def\Aut{ {\rm Aut} }

\def\cR{\mathcal{R}}

\newenvironment{nohyphens}{%
	\hyphenpenalty=10000
	\exhyphenpenalty=10000
	\sloppy %
}{\par}

\title{ Critical dimensions and small  cycle dominance  from all-orders asymptotics of $d$-matrix theory}

\author[a]{Yang Lei}
\author[b]{\!\!, Sanjaye Ramgoolam}

\affiliation[\,a]{Institute for Quantum Science \& School of Physical Science and Technology,  \\ Soochow University, Suzhou 215006, P.R.~China}
\affiliation[\,b]{Centre for Theoretical Physics, Department of Physics and Astronomy, \\Queen Mary University of London, UK}

\emailAdd{leiyang@suda.edu.cn}
\emailAdd{s.ramgoolam@qmul.ac.uk}

\abstract{\begin{nohyphens}
Supersymmetric sectors of $ \cN =4$ super-Yang-Mills theory motivate the study of the partition function for the counting of gauge-invariant functions of  $d=2,3$ matrices transforming under the adjoint action of $U(N)$. The partition function $ \mathcal{Z}_d ( x) $ in the large $N$ limit has a  known Hagedorn phase transition at $ x = d^{ -1} $ which provides a simple model for the phase structure of the thermal partition function of SYM. We study the all-orders asymptotic expansion of  $ \mathcal{Z}_d (x) $ based on a  geometric picture of concentric circles of poles in the complex plane accumulating in a natural boundary at $|x| =1$. We find that the order by order structure has a precise combinatorial interpretation organized in terms of increasing cycle size of permutations arising in the enumeration of the invariants. We refer to this organization as small-cycle dominance, and find that it extends to refined versions of the partition functions depending on several complex variables. An analysis of the coefficients in the asymptotic expansion of $ \mathcal{Z}_d ( x ) $ using the modular property of the Dedekind eta function reveals that the asymptotic expansion is actually convergent for $d\ge d_{ \rm  crit } =  13$.  A fermionic version of $ \mathcal{Z}_d ( x ) $ has an analogous critical dimension of $ d_{ \rm crit} = 7$. 
This distinction indicates that the partition functions  of the  matrix models can be completely reconstructed from their high-energy (UV) limit for $d\ge d_{ \rm crit}$ whereas additional  input is required to reconstruct the exact coefficients of the low-energy (IR) expansion for $2\le d \le d_{ \rm crit } -1 $. 
	\end{nohyphens}
}

\date{}

\begin{document} 

\hfill{QMUL-PH-26-13}

	\maketitle
	
\section{Introduction}\label{sec:intro}

The AdS/CFT correspondence \cite{Maldacena:1997re, Gubser:1998bc,Witten:1998qj} stands as a cornerstone of modern theoretical physics, providing a profound holographic duality between $\mathcal{N}=4$ super-Yang-Mills (SYM) theory with a $U(N)$ gauge group and type IIB string theory in the $AdS_5 \times S^5$ spacetime. 
It offers a remarkable dictionary that translates intractable non-perturbative phenomena in quantum gravity into well-defined computations of observables within quantum field theories and their lower-dimensional quantum mechanical reductions. By systematically analyzing these field-theoretic observables—such as partition functions and the spectra of BPS operators—in the large $N$ limit and beyond, one can extract precise information regarding the quantum physics of gravitons, extended string states, and D-branes in the bulk geometry. 
Consequently, this holographic framework provides a rigorous mathematical testing ground for understanding the dynamics of extended objects and black hole microstates.

A consequence of the holographic dictionary is the correspondence between the thermal phase structures in quantum gravity and the dual gauge theory. 
Specifically, the Hawking-Page phase transition—a geometric transition separating a thermal gas in Anti-de Sitter space from a stable macroscopic black hole—maps  to the confinement/deconfinement phase transition of the gauge theory formulated on a compact manifold \cite{HawkPage,witten2}, 
where the exponential proliferation of gauge-invariant operators manifests as a Hagedorn transition. 
The core physics of this deconfinement/Hagedorn transition is captured by a simpler framework: the $U(N)$ gauged quantum mechanics of multi-matrix models governed by a harmonic oscillator potential \cite{Sundborg:1999ue,Aharony:2003sx}.
Within this context, the two-matrix system plays a  distinguished role, as its quantum states are in  one-to-one correspondence with the quarter-BPS operators belonging to the $SU(2)$ scalar sub-sector of free $\mathcal{N}=4$ SYM \cite{Bianchi:2006ti}. 
Understanding the exact degeneracy of this two-matrix system therefore provides a highly tractable yet rigorous window into the stringy and thermal properties of the dual geometry.

Beyond capturing macroscopic phase transitions, multi-matrix quantum mechanical models—also referred to as $d$-matrix theories—serve as fruitful laboratories for dissecting the fine-grained spectrum and dynamics of the AdS/CFT correspondence.
For instance, these models encode the integrability properties of $\mathcal{N}=4$ SYM \cite{Beisert:2004ry} in the $N=\infty$ planar limit, effectively translating the complex operator mixing problem into tractable spin-chain Hamiltonians.
Furthermore, $d$-matrix frameworks provide the essential algebraic scaffolding required to construct orthogonal bases for the Hilbert space which are critical for describing the finite-$N$ dynamics and mixing of giant graviton fluctuations \cite{Corley:2001zk,KR,BHR1,Rob1,Rob2,BHR2,EHS}. 
They also allow for the systematic isolation of near-BPS closed sub-sectors, enabling the study of non-relativistic, strongly coupled corners of the holographic dictionary with analytic controls \cite{Harmark:2007px,Harmark:2014mpa,Harmark:2019zkn,Baiguera:2020jgy,Baiguera:2020mgk,Baiguera:2021hky,Baiguera:2022pll}.
Furthermore, a zero-charge sector of the 2-matrix quantum mechanics displays negative specific heat capacity \cite{OConnor:2024udv}, which is of interest in the context of small black holes in AdS \cite{Han16,Berenstein:2018lrm}.

The $SU(2)$ subsector of $ \cN=4$ SYM  with $U(N)$ gauge group  consists of two  $ N \times N$ complex matrix scalars $X,Y$ transforming in the adjoint of  gauge group. The counting of the polynomial holomorphic gauge invariants, which are quarter BPS in the free gauge theory limit,  can be organised by the degrees $m,n$ in the two matrices. For $m +n \le N$, the counting is independent of $N$ and the two-variable generating function is
\cite{Dolan:2007rq}
\begin{equation}\label{eq:refined-coefficnet-partition}
	\mathcal{Z}(x,y)= \prod_{i=1}^\infty \frac{1}{1-x^i -y^i}  \,.
\end{equation}
This counting is also applicable for polynomial gauge invariants of two hermitian matrices. In the unrefined limit where $x=y$, this generating function reduces to the $d=2$ case of a broader class of $d$-matrix partition functions \cite{willenbring2007stable}:
\begin{equation}\label{eq:partition-unrefine}
\mathcal{Z}_d(x) =  \frac{1-d}{(d;x)_\infty} = \prod_{i=1}^{\infty} \frac{1}{1-dx^i} = \sum_{K=0}^\infty Z_d(K) x^K \,,
\end{equation}
where  we have written the partition function in terms of $q$-Pochhammer symbol. This gives the large $N$ counting of homolorphic polynomial gauge invariants of $d$ complex matrices or alternatively the large $N$ counting of general polynomial gauge invariants of $d$  hermitian matrices, transforming in the adjoint of $U(N)$. 
While the weighted partition numbers $Z_2(K)$ have been explored numerically for the first few values—leading to the empirical conjecture that the leading growth scales as $Z_2(K)\sim 3.46\times 2^K$ \cite{oeisA070933}, this sequence remains largely understudied. 
Specifically, there currently exists neither an exact analytic formula for the asymptotic coefficients nor a systematic derivation of their (non)-perturbative subleading corrections.
On the contrary, the best studied example of exact degeneracy computation is the integer partition number $Z_1(K)$, which corresponds to the coefficients of the partition function $\mathcal{Z}_1(x)$.
The asymptotic growth of $Z_1(K)$ is captured by Hardy-Ramanujan-Rademacher formula \cite{241Rademacher}.
This celebrated formula provides an exact, convergent series to compute the integer partition numbers. Until now, an analogous analytic structure for the weighted partition numbers $Z_2(K)$ has remained  unknown.

In this paper, we establish that the degeneracy $Z_2(K)$ admits an all-orders asymptotic expansion at large $K$. 
Specifically, we show that:
\begin{equation}\label{eq:degeneracyZ2-0}
Z_2(K) \sim \sum_{n=1}^\infty \sum_{j=0}^{n-1} c_{n;j} \omega_n^{-jK} 2^{\frac{K}{n}} \,,
\end{equation}
where $\omega_n^j$ are $n$-th roots of unity, and the residue coefficients $c_{n;j}$ are constants to be determined in this paper.   
Other main results can be summarized as follows:
\begin{itemize} 
\item We demonstrate that for any integer $d \ge 2$, the degeneracy $Z_d(K)$ admits an all-orders asymptotic expansion of the form:
\begin{equation}\label{Zdkasymp}  
Z_d ( K ) \sim  \sum_{ n=1}^{ \infty } Z_{d,n} ( K ) \,.
\end{equation}
as the generalization of \eqref{eq:degeneracyZ2-0}.
 following the standard Poincaré definition of asymptotic series. The proof of this result relies on the iterative application of key lemmas in analytic combinatorics \cite{FlajoletSedgewick:2009}, systematically subtracting the singular polar parts of $\mathcal{Z}_d(x)$ localized on concentric circles of increasing radii. 
We also describe the broader applicability of this procedure by extending it to generating functions of weighted partitions $\mathcal{Z} (x; \vec{w}) $ where $w_i$ represents the weight assigned to each part of the partitions.

\item We display a remarkable transition in the analytic properties of the $d$-matrix theory.  
For $d \le 12$, the magnitude of the expansion coefficients $Z_{d,n}(K)$ increases exponentially at large $n$, rendering the series formally divergent (though asymptotically valid). 
However, for $d \ge 13$, these coefficients decrease exponentially, meaning that the asymptotic expansion becomes absolutely convergent in this regime. 
We interpret this as indicating that the all-orders  high-energy (UV) information is sufficient to reconstruct the degeneracy counting function for $d\ge 13$ but needs extra input to resolve the ambiguity for $2\le d\le 12$.

\item 
We provide a  combinatorial interpretation for the expansion order $n$. Building upon known expressions for $Z_d(K)$ as a sum over permutations \cite{Pasukonis:2013ts}, 
we show that the index $n$ in the asymptotic series organizes  contributing partitions, describing cycle structures of permutations, according to their increasing minimal cycle lengths.
This phenomenon, which we term asymptotic small cycle dominance, conceptually mirrors the all-orders asymptotics observed in the counting of tensor model observables \cite{BGSR1,BenGeloun:2017vwn,BenGeloun:2021cuj}.
Furthermore, we demonstrate that this small cycle dominance methodology can be applied to extract the asymptotics of the refined partition function $\mathcal{Z}(x, y)$. An important point is that this small-cycle organisation of the asymptotics refers to the cycle decomposition of {\it equivalence permutations  } which relate different {\it contraction permutations } producing the same multi-trace.  The distinction is explained in \eqref{sec:reviewcount}. The cycles of the contraction permutations are the traces in the matrix description. For fixed cycle structures of the equivalence permutations, there is a distribution of trace structures, which is elucidated  in section \ref{sec:trace-picture}. 

\end{itemize} 

The remainder of this paper is organized as follows. 
In Section \ref{sec:complexnalysis}, we use the pole structure of $ \cZ_2(x )$ within the unit disc to derive the all-orders asymptotic formula for the sequence of coefficients $ Z_2( K ) $  in the Taylor  expansion of $ \cZ_2(x)$
around the origin. Section \ref{ssec:general-d-theory} generalizes this derivation to the 
partition function $ \cZ_d ( x ) $ of the bosonic $d$-matrix theory, and finds the critical value $ d = 13$ above which the asymptotic expansion becomes convergent. We also extend the discussion to  a fermionic $d$-matrix theory finding convergence for $ d \ge 7$. In section \ref{sec:MTconeq}, we provide a combinatorial perspective on these asymptotic expansions.  We identify the specific permutation configurations that contribute to the leading and sub-leading orders of the degeneracy, demonstrating that these contributions are naturally organized by the size of their minimal cycles. 
We also explain how  the small-cycle dominance mechanism extends  to refined partition functions, leaving a more detailed treatment of the connection to multi-variable complex analysis for the future. 
Section \ref{sec:CircSings} extends our analysis to combinatorial partition models with general weight sequences.  
We present concrete examples illustrating a scenario where the small cycle dominance of the kind that applies to the $d$-matrix model extends simply, as well as a case where it must be modified to a variant where small cycle dominance sets in after a re-organisation at very small cycle lengths. In section \ref{sec:trace-picture}, we will describe 
the distribution of trace-structures for fixed cycle structures of the equivalence permutations. This is given in terms of cycle indices of wreath-product groups, whose structure depends on the fixed cycle structures. 
Finally, we conclude with a discussion of physical implications and future directions in Section \ref{sec:discussion}.

\section{All orders asymptotic expansion  of  $\mathcal{Z}_2(x)$}
\label{sec:complexnalysis}

In this section, we apply complex analysis to the unrefined $2$-matrix partition function $\mathcal{Z}_2(x)$.
The growth of the coefficients $Z_2(K)$ is intrinsically linked to the singularity structure of $\mathcal{Z}_2(x)$ on the complex plane. 
According to the foundational principles of singularity analysis \cite{FlajoletSedgewick:2009}
\begin{itemize}
\item \textbf{Principle 1}: The location of the singularities determines the exponential growth rate of the coefficients. Specifically, the dominant singularities—those closest to the origin—define the radius of convergence $R$. 
By Pringsheim's theorem, the coefficients of the Taylor expansion satisfy
\begin{equation}
	[x^n] \mathcal{Z}(x) \sim \frac{1}{R^n} \,.
\end{equation}
For a complete analysis of the relevant contents, see Theorem (IV.7) of \cite{FlajoletSedgewick:2009}.
\item \textbf{Principle 2}: The algebraic or transcendental nature of the singularity (e.g., poles, branch points, or essential singularities) determines the sub-exponential factors, such as power-law corrections or logarithmic fluctuations at large orders.
\end{itemize}

A classic paradigm for extracting asymptotics in partition theory, such as for the integer partition $Z_1(K)$, involves exploiting the modular properties of $\mathcal{Z}_1(x)$ under $\mathrm{SL}(2,\mathbb{Z})$ transformation near essential singularities. 
While $\mathcal{Z}_1(x)$ possesses a natural boundary at $|x|=1$, where essential singularities are dense, the modular symmetry provides a powerful bridge to resolve the divergent behaviour near these points.
In contrast, the 2-matrix partition function $\mathcal{Z}_2(x)$, despite also having a natural boundary at $|x|=1$, exhibits a fundamentally different analytic structure in the interior of the unit disk. 
Its singularities are a sequence of discrete simple poles, including the dominant one. 
This distinct landscape suggests that, rather than relying on global modular symmetries, one can extract the asymptotics by summing the local contributions from these poles—a method akin to a truncated fractional expansion.

To circumvent the complications arising from the natural boundary, we adopt a strategy grounded in the local behaviour of the singularities. By approximating the partition function as a sum of contributions from its dominant poles—effectively a truncated fractional expansion—we can extract the degeneracy $Z_2(K)$
with remarkable precision. In the following, we demonstrate that this approach not only captures the leading asymptotic behaviour but also provides an efficient numerical scheme for higher order corrections.

\subsection{Asymptotics of unrefined $2$-matrix partition function}
\label{ssec:unrefined-2-matrix}

Before focusing on the specific $2$-matrix model, let us consider a general partition function $\mathcal{Z}(x)$ whose singularities in the complex plane consist of a sequence of discrete simple poles. Suppose these poles are organized into concentric layers according to their moduli, where the $n$-th layer lies on a circle of radius $r_n$, with $r_1<r_2<\cdots< 1$.
For each layer $n$, let the set of singularities be $\{ x_{n;j}\}_{j=0}^{t_n-1}$ where $t_n$ denotes the number of poles on the radius $r_n$.

We can systematically resolve these singularities by defining a sequence of partial fraction functions:
\begin{equation}\label{eq:general-Fn} 
	\mathcal{F}_n(x) = \sum_{j=0}^{t_n-1} \frac{c_{n;j}}{1 - x/x_{n;j}} \,, 
\end{equation}
where the coefficients $c_{n;j}$ are related to the residues of $\mathcal{Z}(x)$ at $x=x_{n;j}$ by 
\begin{equation}
	\text{Res}(\mathcal{Z}(x), x=x_{n;j}) = - c_{n;j} x_{n;j} \,.
\end{equation}
This construction allows for a step-by-step analytic continuation of the partition function. Specifically, the truncated sum
\begin{equation}\label{eq:def-remainder-truncation}
	\mathcal{R}_M(x) = \mathcal{Z}(x) - \sum_{n=1}^M 	\mathcal{F}_n(x)  
\end{equation}
is regular within the disk $|x|<r_{n+1}$.
By iteratively subtracting these pole contributions, one effectively increases the domain of convergence, extending the analytic description of the function towards its natural boundary at $|x|=1$.
This is pictorically dictated in Figure \ref{Fig:singularity}.
\begin{figure}\centering
\begin{tikzpicture}[scale=2]
	\draw[->] (-1.5,0) -- (1.5,0) node[right] {$x$};
	\draw[->] (0,-1.5) -- (0,1.5) node[above] {$y$};
	\draw (0,0) node[below left] {$O$};
	\def\rone{0.3}    
	\def\rtwo{0.6}   
	\draw[thick] (0,0) circle (1);        
	\draw[dashed] (0,0) circle (\rtwo);  
	\draw[dashed] (0,0) circle (\rone);   
	\fill[red] (45:\rone) circle (2pt) node[above right] {$r_1$};
	\fill[blue] (135:\rtwo) circle (2pt) node[above left] {$r_2$};
\end{tikzpicture}
\caption{A meromorphic function with singularities (marked by red and blue points) of the simple poles at radius $r_1,r_2$ satisfying $r_1<r_2$.  Singularity subtraction defined by remainder function $\mathcal{R}_1$ in \eqref{eq:def-remainder-truncation} is regular within the disk of radius $r_2$.}
\label{Fig:singularity}
\end{figure}
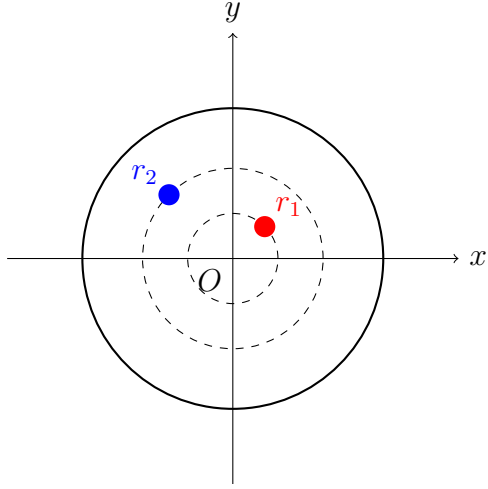

The $2$-matrix partition function $\mathcal{Z}_2(x)$ possesses a sequence of simple poles located at the roots of $1-2x^n=0$ for all positive integers $n$, such that $r_n = 2^{-\frac{1}{n}}$ and $t_n =n$. The poles are explicitly given by
\begin{equation}\label{eq:poles-Z2}
	x_{n;j} = 2^{-\frac{1}{n}} \omega_{n}^j, \qquad \omega_n =e^{\frac{2\pi i}{n}}, \qquad j=0,...,n-1 \,.
\end{equation}
For distinct integers $n$, the sets of poles $x_{n;j}$ are disjoint. 
In the $2$-matrix model under investigation, the sequence of partial fraction functions in \eqref{eq:general-Fn}
has the coefficients
\begin{equation}\label{eq:def-cnj}
	c_{n;j} =	\lim_{x \to x_{n;j}}  	\mathcal{Z}_2(x) (1-2^{\frac{1}{n}} \omega_n^{-j} x ) \,.
\end{equation}
The dominant singularity closest to the origin is $x_{1;0}= \frac{1}{2}$, which dictates the leading exponential growth rate of the coefficients \cite{Ramgoolam:2018epz,oeisA070933}. 
In a physical context, this singularity corresponds to the Hagedorn temperature $T_H =\frac{1}{\ln 2}$ \cite{Atick:1988si,Sundborg:1999ue,Polyakov:2001af,Aharony:2003sx}, a well-known limiting temperature in matrix models and string theory.

While the original Taylor series converges only for $|x|<\frac{1}{2}$, we can systematically extend the domain of the function. 
The remainder $\mathcal{R}_1(x)$ is regular at $x=\frac{1}{2}$ with its convergence radius now limited by the next dominant singularities at $r_2= \frac{1}{\sqrt{2}}$.
By iteratively subtracting the contributions $\mathcal{F}_n(x)$ for all $n$, 
we can analytically continue $\mathcal{Z}_2(x)$ throughout the unit disk. This procedure leads to the following pole expansion:
\begin{align} \label{eq:pole-expansion}
	\begin{split}
		\mathcal{Z}_2 (x) &\sim \sum_{n=1}^\infty  \sum_{j=0}^{n-1} \frac{c_{n;j}}{1-2^{\frac{1}{n}} \omega_n^{-j} x} = \sum_{K=0}^\infty \left[ 
		\sum_{n=1}^{\infty} \sum_{j=0}^{n-1} c_{n;j} \omega_n^{-jK} 2^{\frac{K}{n}}
		\right] x^K \,.
	\end{split}
\end{align}
In the second equality, we expanded the geometric series and collected terms at each order $x^K$. 
Consequently, the degeneracy $Z_2(K)$ at a given energy level $K$ is  given by
\begin{equation}\label{eq:degeneracyZ2}
	\boxed{Z_2(K) \sim \sum_{n=1}^\infty \sum_{j=0}^{n-1} c_{n;j} \omega_n^{-jK} 2^{\frac{K}{n}}} \,.
\end{equation}
This result is consistent with the general correspondence between the local behaviour of a function near its singularities and the asymptotics of its Taylor coefficients (see, e.g., Theorem IV.10 of \cite{FlajoletSedgewick:2009}).
While $T_H=\frac{1}{\ln 2}$ marks the physical Hagedorn temperature—the limiting temperature beyond which the canonical description of the matrix model breaks down—the higher-order poles associated with $T_n= \frac{n}{\ln 2}$ remain indispensable from a mathematical perspective. Even though these 'temperatures' lie beyond the formal domain of the system’s thermal equilibrium, their corresponding singularities encode the necessary sub-exponential corrections required to resolve the exact degeneracy $Z_2(K)$ across all energy scales. 

The first few coefficients $c_{n;j}$ in \eqref{eq:def-cnj} can be determined analytically:
\begin{align}\label{fstfew}
	\begin{split}
c_{1;0} = \frac{1}{\left(\frac{1}{2};\frac{1}{2}\right)_{\infty}}\,, \quad
c_{2;0} =  -\left(\frac{\sqrt{2}+1}{2} \right) \frac{1}{(\frac{1}{\sqrt{2}}; \frac{1}{\sqrt{2}})_{\infty}}\,  \quad
c_{2;1}  = \left(\frac{\sqrt{2}-1}{2} \right)   \frac{1}{( \frac{-1}{\sqrt{2}};  \frac{-1}{\sqrt{2}})_{\infty}}  \,.
	\end{split} 
\end{align}
For $0\le j\le n-1$, representative numerical values of $c_{n;j}$ are summarized in Table \ref{Table:cvalues}.
\begin{table}[H]
	\centering
	\begin{tabular}{|c|c|c|c|c|}
		\hline
$c_{n;j}$		& $j=0$ & $j=1$ & $j=2$ & $j=3$ \\
		\hline
$n=1$	& 3.46  &  &  & \\
		\hline
$n=2$	& -32.17 & 0.222 & & \\
		\hline
$n=3$	& 512.40 & $0.053-0.013 i$ & $0.053+0.013 i$ & \\
		\hline
$n=4$	& -10229 & -0.257 & $0.017-0.01i$ &  $0.017+0.01i $\\
\hline
	\end{tabular}
\caption{Numerical values of the residues  $c_{n;j}$. The column $j=0$  corresponds to the positive real root of  $1-2x^n=0$ for each $n$; the other columns correspond to the complex roots. }
\label{Table:cvalues}
\end{table}
We shall first highlight the pronounced numerical features of these coefficients, while deferring a rigorous analytic investigation of their asymptotic structures to Sections \ref{ssec:analysis} and \ref{ssec:general-d-theory}.
First of all, based on the numerical data in Table \ref{Table:cvalues}, the coefficients $Z_2(K)$ can be formally expressed as
\begin{equation}\label{eq:numer-Z2K-suble}
	Z_2(K) = 3.46 \times 2^K -32.17\times 2^{\frac{K}{2}} +0.222 \times (-1)^K\times 2^{\frac{K}{2}}+ \dots \,.
\end{equation}
The leading term is in excellent agreement with the OEIS sequence A070933 \cite{oeisA070933}, while the subsequent terms provide increasingly refined approximations.
To quantify this, we define the truncated asymptotic expansion of $Z_2(K)$ \eqref{eq:degeneracyZ2} at order $n=M$ as
\begin{equation}\label{eq:Z2-truncation}
	Z_2(K;M) =  \sum_{n=1}^M Z_{2,n}(K) \,.
\end{equation}
The $Z_{2,n}(K)$ are the coefficients the expansion of the partial fraction $\mathcal{F}_n(x)$ at the order $x^K$.
We compare exact $Z_2(K)$ with the truncation at the leading $n=1$, subleading $n=2$ and subsubleading terms $n=3$ in Figure \ref{fig:three-order-difference}, demonstrating the significant improvement in accuracy as more poles are included.
\begin{figure}[]
	\centering
	\includegraphics[width=0.7\linewidth]{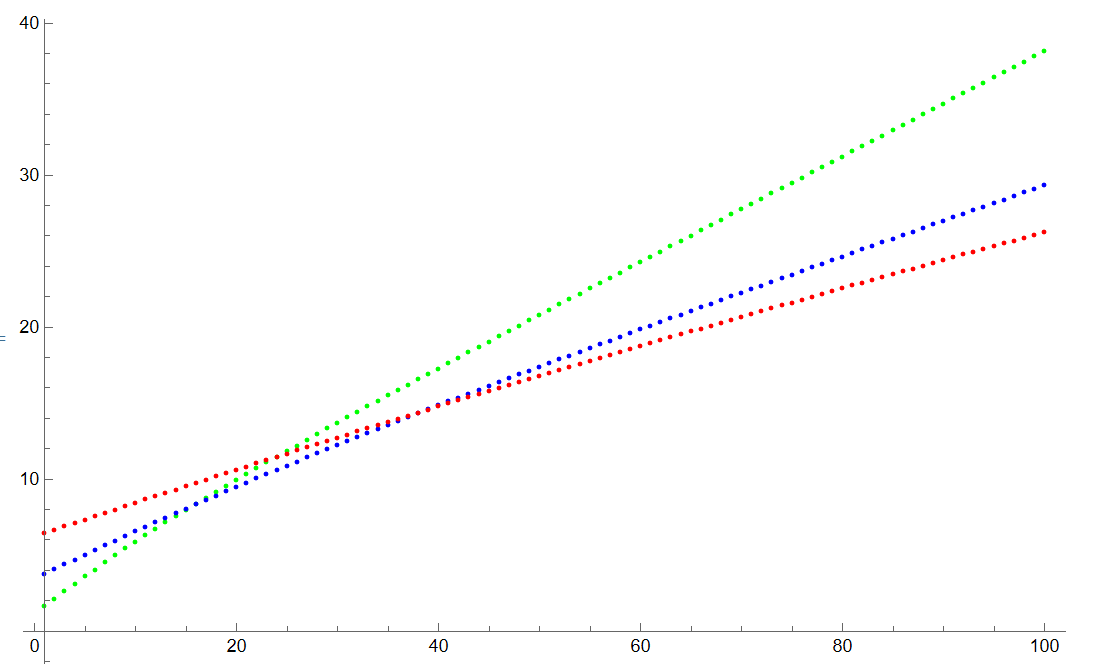}
	\caption{
The vertical axis is log of the error, i.e. $\ln|Z_2(K)- Z_2(K;M)|$, while the horizontal axis is the energy level $K$. 
The green, blue and red correspond to $M=1,2,3$ respectively. 
	}
	\label{fig:three-order-difference}
\end{figure}

We will provide numerical test of \eqref{eq:numer-Z2K-suble} in the appendix \ref{appen:numericalZ2}.
The magnitude $|c_{n;0}|$ associated with the positive real root grows rapidly with $n$, whereas coefficients for complex roots ($j \neq 0$) remain $\mathcal{O}(1)$ or smaller. Furthermore, the sign of $c_{n;0}$ alternates as $(-1)^{n-1}$. Consequently, the first subleading term is negative—a behaviour that stands in sharp contrast to the integer partition $Z_1(K)$ whose non-perturbative subleading corrections  (e.g., the $e^{\pi \sqrt{2K/3}}$ term in the Rademacher expansion \cite{rademacher1937convergent}) consistently possess positive prefactors. 
This sign alternation is deeply rooted in the inclusion-exclusion principle and the specific combinatorial structures explored in Sections \ref{ssec:analysis} and \ref{ssec:general-d-theory}.

Crucially, these subleading terms are non-perturbative: they do not represent standard power-law $1/K$ corrections—which typically correspond to logarithmic entropy corrections in gravity—but rather consist of exponentials with subleading growth rates.
This structure is reminiscent of the "root-of-unity" saddles encountered in black hole microstate counting \cite{Dijkgraaf:2000fq,Manschot:2007ha, Cabo-Bizet:2019eaf,ArabiArdehali:2021nsx,Jejjala:2021hlt,Aharony:2021zkr,Lei:2024oij,Jejjala:2022lrm}.

\subsection{Analytic asymptotic formula for $Z_2(K)$ coefficients}
\label{ssec:analysis}
While $Z_2(K)$ is formally expressed as a double summation over $n$ and $j$, its structure is primarily dictated by the asymptotic growth of the coefficients $c_{n;j}$ in the $n \to \infty$ limit. 
These contributions originate from the poles $1-2x^n = 0$, which accumulate toward the natural boundary at $|x|=1$. 
In terms of the complexified chemical potential $\tau$ (where $x = e^{2\pi i \tau}$), these poles are located at:
\begin{equation}\label{eq:tau-rationas}
	\tau_{n;j} = \frac{i \ln 2}{2\pi n} + \frac{j}{n}, \qquad j, n \in \mathbb{Z},.
\end{equation}
As we shall demonstrate, only the leading coefficients associated with the real positive roots ($j=0$) exhibit rapid growth, while the contributions from complex roots ($j \neq 0$) are exponentially suppressed.

For a fixed integer $n$, the coefficients $c_{n;j}$ defined in \eqref{eq:def-cnj} can be conveniently factorized into two components,   $c_{n;j}^+ (x_{n;j})$ and $c_{n;j}^- (x_{n;j})$, as:
\begin{equation}\label{eq:c-coeff-factor}
c_{n,j} =  \frac{1}{n} c_{n;j}^+ (x_{n;j}) c_{n;j}^- (x_{n;j}) \,. 
\end{equation}
The two factors take the explicit forms:
\begin{align}
	\begin{split}
c_{n;j}^+ (x_{n;j}) &= \left[ \prod_{k=1}^{n-1} (1-2x_{n;j}^k) \right]^{-1}\,, \qquad  \\
c_{n;j}^- (x_{n;j}) &=  \frac{1}{(x_{n;j}; x_{n;j})_{\infty}} = \mathcal{Z}_1(e^{2\pi i \tau_{n;j}})\,.
	\end{split}
\end{align}
The asymptotic behaviour of $c_{n;j}$ is therefore governed by the interplay between these two functions.

We first focus on $c_{n;j}^-$ which is the generating function for Euler partition numbers. 
This function obeys a modular transformation under the $S$-transformation \cite{Cardy:1986ie} (see also \cite{Lei:2024oij} for recent progress on its relation to the modular properties of elliptic Gamma functions):
\begin{equation}\label{eq:modular-dedekind}
\mathcal{Z}_1(e^{-2\pi i \frac{1}{\tau}})  = \frac{1}{\sqrt{\tau}} e^{-i\pi P(\tau)}  \mathcal{Z}_1(e^{2\pi i \tau}), \qquad P(\tau) =\frac{1}{12} \left( \frac{1}{\tau} + \tau -3\right) \,.
\end{equation}
In the regime where $\tau_{n;j} \to 0$ (as $n \to \infty$), the real part of the exponential factor dominates. 
The imaginary part only induces rapid oscillations of order one and does not affect the magnitude.
Utilizing \eqref{eq:modular-dedekind}, the leading asymptotic behaviour of $c^-_{n;j}$ is found to be:
\begin{equation}\label{eq:asymptoc-} 
	\left| c_{n;j}^- (x_{n;j}) \right| \sim \left| \sqrt{\tau_{n;j}} \exp \left( \frac{\pi i}{12 \tau_{n;j}} \right) \right| \sim \frac{(j^2 + \frac{\ln^2 2}{4\pi^2})^{1/4}}{\sqrt{n}} \exp \left[ \frac{n \ln 2}{24 (j^2 + \frac{\ln^2 2}{4\pi^2})} \right] \,. 
\end{equation}

To determine the asymptotics of $c^+_{n;j}$, we rewrite its product representation as an exponential of summation:
\begin{equation}
-\ln c_{n;j}^+ (x_{n;j}) =  \frac{n-1}{2}\ln2+ i\pi (n-1)(j+1) +  \sum_{m=1}^{\infty} \frac{1}{m} \frac{2^{m(\frac{1-n}{n})} -e^{\frac{2\pi i jm}{n}}}{2^{\frac{m}{n}} - e^{\frac{2\pi i j m}{n}}} \,.
\end{equation}
In the large-$n$ limit, this yields:
\begin{equation}\label{eq:asymptoc+}
c_{n;j}^+ (x_{n;j})  \sim (-1)^{(n-1)(j+1)} \exp \left[
n \frac{\ln 2}{48}    \frac{1-24 j^2}{j^2 + (\frac{\ln^2 2}{4\pi^2})} 
\right] \times \text{(phase)} \,.
\end{equation}
Combining \eqref{eq:asymptoc-} and \eqref{eq:asymptoc+}, the asymptotic growth of $c_{n;j}$ is summarized as:
\begin{equation}\label{eq:cnj-asymptotic}
	c_{n;j}(x_{n;j}) \sim \frac{(j^2+\frac{\ln^2 2}{4\pi^2})^{\frac{1}{4}} }{n^{\frac{3}{2}}} (-1)^{(n-1)(j+1)} \exp\left[\frac{1-8j^2}{16} \frac{n\ln 2}{j^2+ \frac{\ln^2 2}{4\pi^2}}  \right]  \times \text{(phase)}  \,.
\end{equation}
This formula reveals the distinct roles of the roots $x_{n;j}$: 
\begin{itemize}
\item \textbf{Real Positive Roots} ($j=0$):  The coefficient grows exponentially as
\begin{equation}
c_{n;0} (x_{n;0})\sim \sqrt{\frac{\ln 2}{2\pi}} \frac{(-1)^{n-1} }{ n^{\frac{3}{2}}}\exp\left(\frac{\pi^2}{4\ln2}n\right) \,.
	\end{equation}
in excellent agreement with Table~\ref{Table:cvalues}.
\item \textbf{Complex Roots}  $(j=\{1,\cdots,n-1\})$. 
Since $\frac{\ln 2}{2\pi}\sim 0.11 \ll j$ and $ 8j^2 \gg 1$, the exponential factor becomes a suppression term.
The decay rate is approximately
\begin{equation}
	|c_{n;j}(x_{n;j})|\sim \frac{\sqrt{j}}{n^{\frac{3}{2}}} \exp\left(- \frac{\ln 2}{2} n\right)\,, \qquad j\neq0  \,,
\end{equation}
as confirmed numerically.
\end{itemize}
Consequently, the summation for $Z_2(K)$ is dominated by the contributions from real roots.

\section{Critical dimensions in $d$-matrix theory}
\label{ssec:general-d-theory}

The $2$-matrix theory governing the $SU(2)$ subsector of $\mathcal{N}=4$ SYM naturally generalizes to larger closed subsectors. For instance, the $SU(2|3)$ sector constitutes a decoupled $1/8$-BPS subsector comprising three scalars and two chiral fermions \cite{Beisert:2004ry,Harmark:2007px}. To develop the analytic tools required to systematically study these extended configurations, this section generalizes our asymptotic framework to bosonic and fermionic $d$-matrix theories respectively.

\subsection{Bosonic $d$-matrix theory}
\label{sec:cridbos}

Our analysis extends naturally to the unrefined limit of the $d$-matrix theory \cite{Berenstein:2018lrm,Rob1,OConnor2023TraceRelations} ($d>1$) whose partition function is given by:
\begin{equation}\label{eq:d-matrix-partition}
	\mathcal{Z}_d(x) =\prod_{i=1}^\infty \frac{1}{1-d x^i}  = \sum_{K=0}^\infty Z_d(K) x^K \,.
\end{equation}
These models represent the $N=\infty$ limit of $d$ free scalars transforming in the vector representation of O$(d)$ \cite{Aharony:2003sx}. 

As with the $d=2$ case studied in Sections \ref{ssec:unrefined-2-matrix} and \ref{ssec:analysis}, the singularities of \eqref{eq:d-matrix-partition} are simple poles located at:
\begin{equation}
x_{n;j} = d^{-\frac{1}{n }} \omega_{n}^j, \qquad j=0, \cdots,n-1 \,.
\end{equation}
Applying partial fraction decomposition, the degeneracy $Z_d(K)$ can be expressed as an asymptotic series with coefficients $Z_{d,n} (K) $ 
given below 
\begin{equation}\label{eq:general-d-subleading-ZK}
Z_d(K) \sim \sum_{n=1}^\infty Z_{d,n} (K) =  \sum_{n=1}^\infty \sum_{j=0}^{n-1} c_{n;j} \omega_n^{-jK} d^{\frac{K}{n}} \,.
\end{equation}
This representation implies the following error bound for the truncated sum:
\begin{equation}\label{eq:asymptotic-ZdK-BigO}
	\left|Z_d(K) - \sum_{n=1}^M \sum_{j=0}^{n-1} c_{n;j} \omega_n^{-jK} d^{\frac{K}{n}} \right| \le \mathcal{O} \left(\left|\sum_{j=0}^{M} c_{M+1;j} \omega_{M+1}^{-jK} d^{\frac{K}{M+1}} \right| \right) \,.
\end{equation}
In specific parameter regimes where the right-hand side of \eqref{eq:asymptotic-ZdK-BigO} becomes infinitesimally small as $M \to \infty$, the asymptotic series \eqref{eq:general-d-subleading-ZK} reduces to a convergent one, a feature we explore below.

The coefficients $c_{n;j}$ are computed following \eqref{eq:def-cnj} and retain the factorized form \eqref{eq:c-coeff-factor}, where:
\begin{align}
	\begin{split}
		c_{n;j}^+ (x_{n;j}) &= \left[ \prod_{k=1}^{n-1} (1-dx_{n;j}^k) \right]^{-1}\,, \qquad c_{n;j}^- (x_{n;j}) =  \frac{1}{(x_{n;j}; x_{n;j})_{\infty}} \,,\\
x_{n;j} &= e^{2\pi i \tau_{n;j}}, \qquad \tau_{n;j} = \frac{\ln d}{2\pi n} i + \frac{j}{n}	 \,.	
	\end{split}
\end{align}
As established in Section \ref{ssec:analysis}, the large-$n$ asymptotics of $c_{n;j}^{-} (x_{n;j})$ are determined by the modular properties of the function $\mathcal{Z}_1$ function, while those of $c^+_{n;j}(x;j)$ follow from an exponential summation formula. 
The resulting asymptotic expansions for these factors are:
\begin{align}
	\begin{split}
 |c_{n;j}^+(x_{n;j})| &\sim \exp\left[
-\frac{n \ln d }{24 \pi^2 (j^2+ \frac{\ln^2 d}{4\pi^2})} \Big(12j^2\pi^2 +F(d)\Big)
\right] \,, \\
\left|c_{n;j}^- (x_{n;j}) \right| &\sim \frac{(j^2 + \frac{\ln^2 d}{4\pi^2})^{1/4}}{\sqrt{n}} \exp\left[	\frac{n \ln d}{24 (j^2+ \frac{\ln^2 d}{4\pi^2})} \right]  \,.
	\end{split}
\end{align}
where $F(d)$ is defined as:
\begin{equation}\label{eq:def-F(d)}
F(d) = 3\ln^2 d+6\text{Li}_2\left(\frac{1}{d}\right) -\pi^2 \,,
\end{equation} 
satisfying $F(1)=0$ and Li$_2$ denotes the dilogarithm function \cite{zagier2007dilogarithm} . 
Consequently, the full coefficient $c_{n;j}$ behaves as 
\begin{equation}\label{eq:cnj-generald}
|c_{n;j}| \sim \frac{(j^2 + \frac{\ln^2 d}{4\pi^2})^{1/4}}{n^{\frac{3}{2}}}	\exp\left[ 
\frac{n \ln d }{24 \pi^2 (j^2+ \frac{\ln^2 d}{4\pi^2})}  \Big( \pi^2-F(d)  - 12j^2 \pi^2 \Big)
\right] \,.
\end{equation}

Several remarks are in order. First, for general $d$, the contribution from the positive real root is not necessarily exponentially growing. 
For sufficiently large $d$, even the $c_{n;0}$ term becomes exponentially suppressed if  $F(d)>\pi^2$.
Since $\ln^2 d$ is a monotonically increasing function,  
there exists a critical value $d_{\star}\approx 12.59$ above which this suppression occurs.
This value can be approximately estimated by neglecting the dilogarithm term $\text{Li}_2\left(\frac{1}{d}\right)$ as it is incomparable to other terms as $d$ gets large. 
The approximation value is then \footnote{This exponential factor notably appears in the Hardy-Ramanujan formula for the partition function \cite{HardyAsymptoticFI}
\begin{equation}
Z_1(K) \sim \frac{1}{4\sqrt{3}K} e^{\sqrt{\frac{2 K}{3}} \pi }
	\end{equation}} 
\begin{equation}
	3\ln^2 d \approx 2\pi^2\quad  \Rightarrow \quad d\approx e^{\sqrt{\frac{2}{3}} \pi} \approx 13.002 \,.
\end{equation}
For $d \ge d_\star$, the exponential suppression of $c_{n;0}$ as $n \to \infty$ implies that only the first few terms in \eqref{eq:general-d-subleading-ZK} contribute significantly, suggesting that the series may become convergent in this regime.
As a result,  the summation $\sum_{j=0}^{n-1} c_{n;j}$ is a convergent series as $c_{n;0}$ is exponentially suppressed in the large $n$ for $d\ge 13$.
Therefore,	the series expansion $Z_d(K)$ in \eqref{eq:general-d-subleading-ZK} is convergent for $d\ge 13$. As we will show in proposition \ref{prop:absconv-dge13} it is absolutely convergent with a uniform convergence property (proposition \ref{prop:UnifConv-dge13}) for discs of radius less than $1$ in the complex plane. 

\begin{proposition}\label{prop:Div-dle12}
For $2 \le d \le 12$ the asymptotic behaviour, the series expansion \eqref{eq:general-d-subleading-ZK} is divergent. 
\end{proposition}

\begin{proof} Equation \eqref{eq:cnj-generald}
 implies that for $  2 \le d \le 12 $ the coefficients $ Z_{d,n} (K)$ grow exponentially with $n$. For the real root $j=0$ in the sum for $ Z_{d,n} (K)$,
\begin{equation}
	|c_{n;0}| \sim n^{-3/2}\exp\!\left[\frac{n}{6\ln d}(\pi^2-F(d))\right].
\end{equation}
Since $F(d)<\pi^2$ in this regime, the magnitude of $c_{n;0}$ increases exponentially with $n$. By contrast, when $j\neq0$ the exponent in \eqref{eq:cnj-generald} contains an additional negative contribution proportional to $12\pi^2 j^2$, implying exponential suppression of the corresponding coefficients at large $n$. 
Although the coefficient at level $n$ involves a sum over all $j$, the number of terms grows only linearly with $n$ while each $j\neq0$ contribution decays exponentially. 
Hence the sum over complex roots cannot cancel the exponentially growing real-root contribution. The dominant behaviour therefore arises from $c_{n;0}$, implying exponential growth of the coefficients with $n$, and hence the expansion \eqref{eq:general-d-subleading-ZK} diverges for $2 \le d \le 12$.
\end{proof}

Furthermore, numerical tests reveal that the monotonicity of $c_{n;0}$ varies with $d$. For $d \le 4$, $c_{n;0}$ increases monotonically with $n$. However, in the range $5 \le d \le 12$, $c_{n;0}$ initially decreases, reaches a minimum at a specific order $n_{\text{min}}$, and then increases. This intermediate region is where the exponential approximations in \eqref{eq:cnj-generald} are most applicable. The values of $n_{\text{min}}$ are summarized in Table \ref{table:Numerical-d-po-re}.
\begin{table}[H]
	\centering
	\begin{tabular}{|c|c|c|c|c|c|c|c|c|c|c|}
		\hline
	$d$	& $d\le 4$ & 5 & 6 & 7 & 8 & 9 & 10 & 11 & 12 &  $d\ge 13$ \\
		\hline
	$n_{min}$& - & 2 & 2 & 3 & 4 & 6 & 8 & 15 & 42 &  $+\infty$ \\
		\hline
	\end{tabular}
	\caption{
In the table we list the values of each $n_{min}$ which indicates $c_{n;0}$ reaches local minimal value for given $d$. 
For $d> d_\star$, the $c_{n;0}$ decreases monotonically. }
	\label{table:Numerical-d-po-re}
\end{table}

The scalars in the above  model transform in the vector representation of $O(d)$. 
It is interesting to note that analogous tensor models involving both $O(N)$ adjoint and $O(d)$ vector representations have been studied extensively \cite{Ferrari:2017ryl, Ferrari:2017jgw, Azeyanagi:2017mre, Carrozza:2020eaz}
 A double scaling limit of large $N,d$ of  the multi matrix model is discussed in \cite{Bonzom:2022yvc}. 
 Interestingly the large-dimension limit has also been considered in  general relativity \cite{Emparan:2013moa, Emparan:2020inr}.
Specifically, the critical dimension $d=12$ also arises in that context: black string to black hole phase transitions exhibit a critical threshold at $D=d+1 \approx 13.5$, distinguishing first order from higher-order Gregory-Laflamme instabilities \cite{Sorkin:2004qq, Suzuki:2015axa, Kol:2006vu, Harmark:2005pp, Kol:2004ww}. 
We will discuss this further in Section \ref{sec:discussion}.

The exponential decay of $|c_{n;j}|$ at large $n$ leads to the following propositions, which characterize the convergence and the analytic structure of the $d$-matrix partition function for $ d \ge 13$.  
\begin{proposition}\label{prop:absconv-dge13}
The series representation for $Z_d(K)$ in \eqref{eq:general-d-subleading-ZK} is absolutely convergent for $d\ge 13$.
\end{proposition}
\begin{proof}
The asymptotic expansion of $c_{n;j}$ in \eqref{eq:cnj-generald} implies the existence of an integer $N_1$ and a positive constant $A$ such that for all $n>N_1$ \footnote{Numerical test shows the $N_1$ to achieve this exponential asymptotic approximation for $d=13$ could be of order $10^4$. For larger $d$, the exponential decay will be faster and the value of $N_1$ will be smaller. }
\begin{equation}\label{eq:M-cn;j-d}
	|c_{n;j}| \le  \frac{A}{n^{\frac{3}{2}}}	\exp\left[ 
	\frac{n \ln d }{24 \pi^2 (j^2+ \frac{\ln^2 d}{4\pi^2})}  \Big( -F(d) + \pi^2 - 12j^2 \pi^2 \Big)
	\right] \,.
\end{equation}
Since the right-hand side of \eqref{eq:M-cn;j-d} is a monotonically decreasing function of $j$, the principal coefficient ($j=0$) provides a universal upper bound:
\begin{equation}
 	|c_{n;j}|	 \le|c_{n;0}| \le  \frac{A}{n^{\frac{3}{2}}}	\exp\left(
 -\lambda_d n
\right), \qquad \lambda_d = \frac{F(d)-\pi^2}{6 \ln d} \,,
 \end{equation}
where $\lambda_d>0$ for $d\ge 13$. 
The series \eqref{eq:general-d-subleading-ZK} is then bounded by:
\begin{equation}\label{eq:ZdK-absolutebound}
\left|\sum_{n=1}^\infty \sum_{j=0}^{n-1} c_{n;j} \omega_n^{-jK} d^{\frac{K}{n}}\right| \le \sum_{n=1}^{N_1-1} \left|
\sum_{j=0}^{n-1} c_{n;j} \omega_n^{-jK} d^{\frac{K}{n}}
	 \right|  + \sum_{n=N_1}^\infty \tilde{A} \exp(	-\lambda_d n) \,,
\end{equation}
where the $n$-independent constant $\tilde{A}$ is
\begin{equation}
	\tilde{A} =  A	\exp\left(
	\frac{K}{N_1} \ln d - \frac{1}{2} \ln N_1
	\right) \,.
\end{equation}
 is an $n$-independent constant. Since the second term on the right-hand side is a convergent geometric-like series, \eqref{eq:general-d-subleading-ZK} is absolutely convergent.
\end{proof} 

We now establish a uniform convergence result for $ \cZ_d (x)$ in the regime $ d \ge 13$. 
\begin{proposition}\label{prop:UnifConv-dge13}
For $d\ge 13$, the partial fraction summation 
\begin{equation}\label{eq:sum-residue-Z}
S_M(x)= 1+\sum_{n=1}^M \sum_{j=0}^{n-1}\left( \frac{c_{n;j}}{1-d^{\frac{1}{n}} \omega_n^{-j}x} d^{\frac{1}{n}} \omega_n^{-j} x\right)
\end{equation}
 converges uniformly to a limit function $S_{d;\infty}(x)$ on any compact subset of the unit disk $|x|<1$.
\end{proposition}
\begin{proof}
Consider the partial fraction sum $S_{d;M}(x)$ \eqref{eq:sum-residue-Z}.
For any fixed radius $r<1$, we can choose an integer $M_0$ such that $r<d^{\frac{-1}{n}}$ for all $n>M_0$, ensuring all poles $x_{n;j}$ with $n>M_0$ lie outside the disk $|x|<r$. 
To establish uniform convergence, we examine the remainder for $M>M_0$:
\begin{equation}
\label{eq:bound-on-Zd-function} 
|S_{d;\infty} (x) - S_{d;M}(x)| \le \sum_{n=M+1}^\infty\frac{n \cdot r \cdot |c_{n;j}|}{d^{-1/n} - r}\,. 
\end{equation}
As $n \to \infty$, the term $d^{-1/n} - r$ approaches $1-r > 0$. 
The convergence of the tail is therefore dictated by the exponential decay of $|c_{n;0}|$. 
This ensures that $S_{d;M}(x)$ converges uniformly to a meromorphic function $S_{d;\infty}(x)$. 
\end{proof}

We therefore observe a transition between divergent series and convergent series parametrized by critical dimension. 
This phenomenon should be an $N=\infty$ effect and will not be present for any finite but large $N$ \cite{Harmark:2014mpa,Dolan:2007rq}. 
For any finite $N$, no matter how large, the partition function $\mathcal{Z}_d(N,x)$ only has poles at the boundary of the unit disk instead of layers of poles studied in this paper. 
The zeros of the partition function approach to the Hagedorn temperature by Lee-Yang type singularity \cite{Kristensson:2020nly}.

For $d\ge 13$, Proposition~\ref{prop:UnifConv-dge13} shows that the circles-of-poles expansion converges uniformly on compact subsets of the open unit disk to a holomorphic function $S_{d;\infty}(x)$. A remaining question is whether this limit agrees exactly with the original generating function $\mathcal{Z}_d(x)$. The issue is that the pole expansion  detects only poles strictly inside the unit disk. For example, 
\begin{equation}
 \tilde{\mathcal{Z}}_d(x)=\mathcal{Z}_d(x)+\frac{1}{1-x}
\end{equation}
has the same interior pole structure as $\mathcal{Z}_d(x)$, since $(1-x)^{-1}$ has no poles on any  layer $|x|=d^{-1/n}<1$. Hence $(1-x)^{-1}$ is invisible to the circles-of-poles analysis, even though it changes the holomorphic function on $|x|<1$. Thus, to prove $S_{d;\infty}(x)=\mathcal{Z}_d(x)$, one must rule out the addition of a function holomorphic in the open unit disk whose singular behaviour appears only on the boundary $|x|=1$. Our numerical evidence strongly supports that is no such extra term in  the present case.

We perform explicit numerical evaluations of the coefficients $ \mathcal{Z}_d ( x )$ and  $S_{d;M}(x)$. 
The discrepancies between the exact coefficients of $\mathcal{Z}_d(x)$ and those extracted from the truncated partial fraction sum $S_{d;M}(x)$ are heavily suppressed. 
For example, considering the macroscopic energy $K=70$ at $d=14$, the exact degeneracy $Z_{14}(70)$ is of order $\mathcal{O}(10^{80.26})$. 
In comparison, the absolute errors from the asymptotic approximations \eqref{eq:asymptotic-ZdK-BigO} are negligible even we change the truncation order $M$:
\begin{align}
	\begin{split}
& [x^{70}] \Big|\mathcal{Z}_{d=14}(x) - S_{M=13}(x) \Big| \sim 10^{2.88} \,, \qquad \delta \sim \mathcal{O}(10^{-77.38})\\
& [x^{70}]  \Big|\mathcal{Z}_{d=14}(x) - S_{M=17}(x) \Big| \sim 10^{1.17} \,,\qquad \delta \sim \mathcal{O}(10^{-79.10}) \,.
	\end{split}
\end{align}
where we denote $\delta$ as the ratio between the errors and the scale of $Z_d(K)$.
This striking suppression extends deep into the low-energy regime, where the partial fraction summation becomes almost exact. 
For instance, evaluating at $(K,d,M)=(6,14,17)$ yields a microscopic absolute error:
\begin{equation}
[x^{6}]  \Big|\mathcal{Z}_{d=14}(x) - S_{14,17}(x) \Big| \sim10^{-3.05}	\,,
\qquad  \delta \sim\mathcal{O}(10^{-9.90}) \,.
\end{equation}
Such extreme numerical precision across both low- and high-energy regimes strongly indicates that the convergent sequence $S_{d;\infty}(x)$ reconstructs the original partition function $\mathcal{Z}_d ( x ) $  exactly, leaving no room for additional analytic corrections.
 
 Assuming the equality between $ \mathcal{Z}_d ( x ) $ and the convergent sum $ S_{d;\infty}(x)$ for $ d \ge 13$, which is strongly suggested by the numerical evidence, we note the analogy between 
 \bea 
\mathcal{Z}_d ( x ) = S_{d;\infty}( x )
 \eea
 and the Mittag-Leffler expansion for analytic functions with simple poles  and bounded on the finite complex plane \cite{whittaker2020course}. 
 We may think of the pole expansion of $ \mathcal{Z}_d (x)$ as a generalisation of the Mittag-Leffler expansion which exists for certain classes of meromorphic functions which are not necessarily bounded on the finite complex plane and have a natural boundary at finite radius.

\subsection{Fermionic $d$-matrix theory}
\label{sec:critferm} 

The discovery of a critical dimension $d=12$ in the bosonic sector raises a natural question: is this threshold universal across different matrix models?
To address this, we turn our attention to the fermionic $d$-matrix theory, characterized by the partition function \cite{Aharony:2003sx,Dolan:2007rq,OConnor2023TraceRelations}
\begin{equation}\label{eq:partitionfunction-fermioninc-d}
	\mathcal{Z}_d^F (x) = \prod_{n=1}^\infty \frac{1}{1+d(-x)^n} \,.
\end{equation}
In this model, the even and odd modes of $n$ exhibit distinct analytic structures, necessitating separate treatments for their respective poles:
\begin{align}
	\begin{split}
		x_{2n-1;j} &= d^{-\frac{1}{2n-1}} \omega_{2n-1}^j, \qquad j=0,\cdots,2n-2 \\
		x_{2n;j} &= d^{-\frac{1}{2n}} \omega_{2n}^{j+\frac{1}{2}}, \qquad j=0,\cdots,2n-1 \,.
	\end{split}
\end{align}
Through partial fraction decomposition, the partition function is decomposed by fractions as
\begin{equation}
\mathcal{Z}_d^F (x) \sim \sum_{n=1}^\infty \sum_{j=0}^{2n-2} \frac{c_{2n-1;j}}{1-d^{\frac{1}{2n-1}} \omega_{2n-1}^{-j} x} + \sum_{n=1}^\infty \sum_{j=0}^{2n-1}  \frac{c_{2n;j}}{1-d^{\frac{1}{2n}} \omega_{2n}^{-j-\frac{1}{2}} x}  \,,
\end{equation}
where the $1/2$-shift in the index $j$ for the even-sector coefficients $c_{2n;j}$ arises from the $(-1)^n$ factor in the fermionic partition function \eqref{eq:partitionfunction-fermioninc-d}.
We denote $Z_d^F(K)$ as the Taylor coefficient of this expansion.

While the computation of $c_{2n-1;j}$ and $c_{2n;j}$ is more technically involved than in the bosonic case, it follows a parallel logic. 
For instance, $c_{2n-1;j}$ can be expressed as:
\begin{align}
	\begin{split}
		c_{2n-1;j} = \frac{1}{n} \prod_{k=1}^{n-1} \left[ \frac{1}{1+d x_{2n-1;j}^{2k}}  \frac{1}{1-dx_{2n-1;j}^{2k-1}} \right]  \prod_{k=1}^{\infty} \frac{1}{1+x_{2n-1;j}^{2k-1}} \prod_{k=1}^\infty \frac{1}{1-x_{2n-1;j}^{2k}} \,.
	\end{split}
\end{align}
To extract the large-$n$ asymptotics, we utilize the identity for the alternating product:
\begin{equation}
	\prod_{k=1}^\infty \frac{1}{1+x^{2k-1}} = \frac{1}{(-x;x^2)_{\infty}} = \frac{\mathcal{Z}_1(x) \mathcal{Z}_1(x^4)}{\mathcal{Z}_1(x^2)^2} \,,
\end{equation}
which allows us to apply modular transformations \eqref{eq:modular-dedekind} to the $\mathcal{Z}_1$ factors.
The first square bracket factor can be computed in the large $n$ limit by its exponential summation formula. 

The leading-order behaviour of the coefficients is governed by the functional:
\begin{equation}
	\mathcal{F}(d;j) = 6 \ln^2 d +3 \text{Li}_2 \left(\frac{1}{d^2}\right) + 2\pi^2(12j^2-1) \,.
\end{equation}
The zero of this function at $j=0$ is $d\approx 6.111$.
Specifically, the coefficients in the large-$n$ limit behave as:
\begin{align}
	\begin{split}
c_{2n;j} &\sim \exp\left[-\frac{n \ln d}{24\pi^2} \frac{\mathcal{F}(d;j+\frac{1}{2})}{(j+\frac{1}{2})^2 + \frac{\ln^2 d}{4\pi^2}} \right] \,, \\
c_{2n-1;j} &\sim \exp\left[-\frac{n \ln d}{24\pi^2} \frac{\mathcal{F}(d;j)}{j^2 + \frac{\ln^2 d}{4\pi^2}} \right]  \,.
	\end{split}
\end{align}
Analysis of $\mathcal{F}(d;j)$ reveals that for the principal odd-sector root ($j=0$), the exponent remains negative for small $d$, specifically when $d \le 6$. 
Consequently, $c_{2n-1;0}$ grows exponentially in this regime. Conversely, the even-sector coefficients $c_{2n;j}$ always decay regardless of $d$ or $j$; the $1/2$-shift in the angular index ensures that $\mathcal{F}(d; j+1/2)$ is strictly positive. As both sectors contribute to the total degeneracy $Z_d^F(K)$, the stability of the series is dictated by the odd sector, yielding a fermionic critical dimension of $d_\star^F = 6$.

Consistent with the bosonic model, this threshold can be estimated by neglecting the dilogarithm term:
\begin{equation}
	6\ln^2 d \approx 2\pi^2 \quad \Rightarrow \quad d \approx e^{\pi/\sqrt{3}} \approx 6.1 \,.
\end{equation}
In conclusion, both bosonic and fermionic $d$-matrix models possess distinct critical dimensions $d_\star$. 
For $d \ge d_\star$, the models admit convergent series expansions for both the partition functions and their corresponding degeneracies, signalling a transition in the analytic behaviour of the theory.

\section{Small cycle dominance and asymptotics of multi-traces }
\label{sec:MTconeq}

In this section we review the combinatorial derivation of the refined (and unrefined) counting function for the 2-matrix invariants using an approach based on enumerating the equivalence classes of the permutations of indices which can be used to construct the invariants. This algebraic combinatorial approach extends to the construction of orthogonal bases, the computation of correlators (see  \cite{Corley:2001zk,KR,BHR1,Rob1,Rob2,BHR2,EHS,Padellaro:2024rld} as background and reviews \cite{SR-SWrev,Ramgoolam:2016ciq,deMelloKoch:2024sdf,NegSHCRev} for applications beyond enumeration). A key role is played by the cycle structures of the equivalence generating permutations arising in this description. We show that the all orders asymptotic expansion is organised in terms of cycles of increasing length starting from the smallest. We refer to this phenomenon as small cycle dominance of the asymptotics. The detailed discussion is given in the case of the unrefined partition function, while the extension to the refined case is motivated as an interesting and promising avenue for the future. 

\subsection{Counting multi-traces: contraction and equivalence permutations  }
\label{sec:reviewcount} 

In the $U(N)$ gauge theory, local gauge-invariant operators naturally manifest as multi-traces. 
Any specific multi-trace structure is generated by contracting the fundamental indices in a particular configuration. 
Consider invariants which have degree $m$ in $X$ and degree $n$ in $Y$. The index contractions are pairings between upper and lower indices
\bea 
X^{ i_1}_{ j_1 } X^{ i_2}_{ j_2 } \cdots X^{ i_m  }_{ j_m  }  Y^{ i_{m+1}   }_{ j_{ m+1}   }  \cdots Y^{ i_{ m+n} }_{ j_{ m+n} }  \,.
\eea
Any contraction can be parametrised by a permutation $ \sigma \in  S_{ m+n} $. We denote the corresponding invariant  polynomial 
\bea\label{Osigma}  
\cO_{ \sigma } ( X , Y ) = \sum_{ i_1 , \cdots , i_m , i_{m+1} , \cdots , i_{ m+n} =1}^{ N }  X^{ i_1}_{ i_{ \sigma(1)} } \cdots X^{ i_m}_{ i_{ \sigma(m)} }  Y^{ i_{m+1}   }_{ i_{  \sigma ( m+1 ) }   }  \cdots Y^{ i_{ m+n} }_{ i_{ \sigma (  m+n ) } }  \,.
\eea 
We will refer to $ \sigma $ as an index-contraction permutation, and $S_{ m+n} $ as the group of index-contraction permutations. 
Different permutations  do not always produce distinct invariant functions; rather, they are subject to permutation gauge transformations among letters of the same species, yielding the exact equivalence:
\bea 
\cO_{ \sigma  } ( X , Y ) = \cO_{  \gamma \sigma \gamma^{-1}   } ( X , Y ) ~\hbox{  for any } 
\gamma \in S_m \times S_n  \,.
\eea

Counting invariant functions amounts to enumerating equivalence classes of index-contraction permutations 
$ \sigma \in  S_{ m+n} $. 
Two permutations $ \sigma , \tilde \sigma \in S_{ m+n} $ are considered equivalent if they are related by a permutation gauge transformation:
\bea\label{permequivs} 
\tilde \sigma  = \gamma \sigma  \gamma^{ -1}  \hbox{ for some } \gamma \in S_{ m} \times S_n  \,.
\eea
We will refer to the permutations $ \gamma $ as equivalence-generating permutations.
By Burnside's Lemma \cite{Wiki-Burnside}, the number of such equivalence classes is equal to the average number of fixed points under the conjugation action of $\gamma$. 
A permutation $\sigma$ is a fixed point if it commutes with $\gamma$, namely $\gamma \sigma \gamma^{-1} = \sigma$, or equivalently, $\gamma \sigma \gamma^{-1} \sigma^{-1} = 1$.

To formalize this, we introduce the delta function on the symmetric group $S_{m+n}$, defined as $\delta(\tau) = 1$ if $\tau$ is the identity permutation, and $\delta(\tau) = 0$ otherwise. 
The partition function $Z(m,n)$ is then precisely the fixed-point average:
\begin{eqnarray}\label{Zmnfxdpt}  
&& Z ( m , n ) = { 1  \over m! n! } \sum_{ \gamma \in S_m \times S_n }  \Big(\hbox{ Number of $\sigma \in S_{m+n}$ fixed by $ \gamma $ } \Big) \cr 
&& =  { 1  \over m! n! } \sum_{ \gamma \in S_m \times S_n } \sum_{ \sigma \in S_{ m+n} } 
 \delta ( \gamma \sigma \gamma^{-1} \sigma^{-1} )  \,.
\end{eqnarray}
Any equivalence-generating permutation naturally decomposes into $ \gamma = \gamma_1 \circ \gamma_2$ where  $ \gamma_1 \in S_m , \gamma_2 \in S_n$. The conjugacy classes of $\gamma_1$ and $\gamma_2$ are uniquely characterized by the integer partitions $p \vdash m$ and $q \vdash n$, respectively. 
The number of group elements belonging to these specific conjugacy classes is determined by the order of their respective centralizers (often denoted as the automorphism group of the partition):
\begin{equation}\label{eq:automorphism-pq}
{ m! \over | \Aut ( p ) |  }   = { m!  \over \prod_i i^{ p_i} p_i! } \,, \qquad   { n! \over | \Aut ( q ) | }  =  { n!  \over \prod_i i^{ q_i} q_i! } \,.
\end{equation}
When embedded into the full symmetric group $ S_{ m+n} = S_K $, the composite permutation $\gamma$ possesses the combined cycle structure $ p \circ q $. 
The number of index-contractions $\sigma \in S_{m+n}$ that satisfy $\sigma \gamma \sigma^{-1} = \gamma$ is exactly the order of the centralizer of $\gamma$ in $S_{m+n}$:
\bea \label{eq:automorphism-pq-2}
| \Aut ( p \circ q ) |  = \prod_{ i } i^{ p_i + q_i } ( p_i + q_i ) !  \,.
\eea 
Substituting these combinatorial weights back into \eqref{Zmnfxdpt}, the summation over individual group elements $\gamma$ reduces into a summation over the partitions $p$ and $q$:
\bea\label{Zmnpqsum}  
Z ( m , n  )  = \sum_{ p \vdash m , q \vdash n  } 
  { | \Aut ( p \circ q ) | \over | \Aut ( p ) | | \Aut ( q ) | }   \,.
\eea 
This algebraic reformulation demonstrates that the partitions $p$ and $q$ encode the cycle structures of the equivalence-generating permutations. 
As we will detail in Section \ref{sec:asympsmall}, the leading large-$(m,n)$ asymptotics are robustly dominated by the small-cycle modes within $p \circ q$.

\subsubsection{ From permutations to infinite product } 
\label{lem} 

This is done in earlier papers e.g. \cite{Pasukonis:2013ts} but useful to recall since it is important in seeing the link between the infinite product and the cycle lengths. 

We first establish that the refined partition function, restricted to permutations with cycle lengths no less than $l$, can be expanded in terms of partition automorphisms as follows:
\bea \label{cZsum} 
\sum_{ m , n } x^m  y^n \sum_{ \substack { p \vdash m , q \vdash n \\ ( p \circ q ) \in \cP_l  ( m + n ) }}  { \Aut ( p \circ q ) \over \Aut ( p ) \Aut ( q ) }  = \prod_{ i = l }^{ \infty } { 1 \over  1 - x^i - y^i }  \,,
\eea
where $\mathcal{P}_l ( K )$ is defined as the set of partitions of $K$ consisting exclusively of parts of size no less than $l$, following \cite{BenGeloun:2021cuj}. 
More concretely, given the subset of $p = [ 1^{ p_1} , 2^{ p_2} , \cdots ] $ where  the smallest $i$ with non-zero $p_i$ is $ i = l $:
\begin{equation}
	\mathcal{P}_l(K) = \{p\vdash K | \, p= [l^{p_l} (l+1)^{p_{l+1}} ... K^{p_K}], \quad p_l \neq 0 \} \,.
\end{equation}
To prove the identity \eqref{cZsum}, we recall the formula \eqref{eq:automorphism-pq} and \eqref{eq:automorphism-pq-2} for the order of the automorphism group of a partition $p = [1^{p_1},\cdots ,K^{p_K}]$, $q$ and also $p\circ q$:
\begin{equation}
\Aut ( p ) = \prod_i i^{ p_i } p_i  ! ~;~ \Aut ( q ) = \prod_i i^{ q_i} q_i!  
~;~ \Aut ( p \circ  q ) = \prod_i i^{ p_i + q_i } ( p_i + q_i ) !  \,.
\end{equation}
Applying this definition to the restricted partitions $p$, $q$, and their union $p \circ q$, the ratio of automorphisms on the left-hand side of \eqref{cZsum} simplifies elegantly to a product of binomial coefficients:
\bea 
&& \cZ ( x , y )  = \sum_{ m , n =0}^{ \infty } x^m y^n 
\sum_{ p \vdash m , q \vdash n } \prod_{i}    { ( p_i + q_i ) ! \over p_i! q_i ! }   =  \sum_{ p_i } \sum_{ q_i} 
x^{ \sum_i i p_i } y^{ \sum_i i q_i  } \prod_{i=1}^{ \infty }    { ( p_i + q_i ) ! \over p_i! q_i ! }   \cr 
&& =  \sum_{ p_i } \sum_{ q_i}  
\prod_{i=1}^{ \infty } 
 x^{ i p_i } y^{ i q_i } { (p_i + q_i ) ! \over p_i ! q_i! }  = \prod_{ i =1}^{ \infty } \sum_{ p_i , q_i  }  x^{ i p_i } y^{ i q_i } { (p_i + q_i ) ! \over p_i ! q_i! } \cr 
&& = \prod_{ i=1 }^{ \infty }  \sum_{ n_i =0}^{ \infty } \sum_{ p_i =0}^{ n_i } 
x^{ i p_i } y^{ i ( n_i - p_i ) }   {  n_i ! \over p_i ! ( n_i - p_i )! }  = \prod_{ i =1  }^{ \infty }  \sum_{ n_i =0}^{ \infty }  ( x^i + y^i )^{ n_i } \cr 
&& = \prod_{ i =1  }^{ \infty }  { 1 \over  1 - x^i - y^i  }  \,.
\eea  
Similar algorithm can be applied for $ \cZ_l ( x , y ) $, the generating function for the ratios of automorphisms when the $ p \circ q $ is constrained to have a minimum cycle length of $l$.
This yields 
\begin{equation}\label{eq:Zlxy-function}
\cZ_l ( x , y ) = \prod_{ i = l }^{ \infty }  { 1 \over  1 - x^i - y^i  } \,. 
\end{equation}

\subsection{Asymptotics and  small cycles of equivalence permutations } 
\label{sec:asympsmall} 

In this section, we will show by casting the counting formula as a summation over partitions $p$, the large-$K$ asymptotic behaviour is dominated by configurations of the shortest possible cycles. 
Specifically, the leading contribution arises from partitions consisting almost entirely of length-$1$ cycles, taking the explicit form $p=[1^{K-m},m]$. This structural mechanism strongly echoes findings in rank-3 tensor invariant models, where identical partition configurations dictate the dominant contributions to 3-index tensor observables \cite{BenGeloun:2021cuj}. More broadly, this underlying principle of small-cycle dominance fundamentally controls the high-temperature expansion across various permutation-invariant matrix and tensor models \cite{OConnor:2024udv,PITQM}.

\subsubsection{Unrefined partition function}
\label{ssec:unrefined-smallcycle}
Following \cite{BenGeloun:2021cuj}, we define $ \mathcal{P}_1 ( K ) $ as the specific subset of partitions of $K$ whose minimal part is exactly $1$. In the standard cycle notation $ p = [ 1^{ p_1}, 2^{ p_2} , \cdots ] $, this corresponds to the strict requirement $p_1 \ne 0$. 
Denoting the set of  all the partitions of $ K $ as $ \mathcal{P}(K )$. 
The full partition space admits the exact disjoint decomposition: 
\bea\label{PartitionsDisjoint} 
\cP( K ) = \cP_{ K } (K ) \sqcup \left ( \bigsqcup_{ l \in \{ 1 , 2 , \cdots , \lfloor { { K \over 2 } \rfloor } \} } 
\mathcal{P}_l  ( K ) \right )  \,.
\eea

The unrefined partition function $\mathcal{Z}_2(x)$ can be geometrically expanded as 
\bea \label{SumPwrsx}
\cZ ( x )           && = \sum_{ K=0 }^{ \infty }  x^K \sum_{ p \vdash K } 2^{ C ( p ) } \,,
\eea 
where $C(p) $ is defined as $ \sum_{ i } p_i $. When $p$ is taken to describe a cycle structure of permutations with $p_i $ cycles of length $i$, $ C (p) $ is the total number of cycles in the permutations with cycle structure $p$. 

To estimate the relative weights of these disjoint sectors in the large-$K$ limit, it is instructive to examine the maximal configurations within each subset contributing to the sum \eqref{SumPwrsx}. 
The absolute maximum is achieved by the completely disconnected configuration $p = [1^K ]\in \cP_1(K)$, which yields a macroscopic combinatorial weight of order
$2^{ K} $ to the sum. 
While for $\mathcal{P}_2(K)$ (assuming even $K$ for simplicity), the leading contribution comes from the uniform partition $p = [2^{K/2}]$, yielding a weight of only $2^{ K/2 }$
Generalizing this trend, the maximal contribution from any higher-order sector $\mathcal{P}_l(K)$ scales parametrically as $2^{K/l}$. 
Because the sequence of maximal weights $2^{K/l}$ is exponentially suppressed as $l$ increases, a strict combinatorial hierarchy emerges: $2^K \gg 2^{K/2} \gg 2^{K/3} \dots$ at large $K$. 
This structure not only schematically matches the asymptotic series of $Z_2(K)$ in \eqref{eq:degeneracyZ2}, but also demonstrates that the full partition function $\cZ(K)$ is governed by the $\mathcal{P}_1(K)$ sector, cementing the physical intuition of small-cycle dominance.

Now we will exploit the systematic  decomposition \eqref{PartitionsDisjoint}. 
\begin{eqnarray} \label{Zsuml}  
\cZ ( x )  && = \sum_{ K } x^{ K } \sum_{ l =1 }^{ K } \sum_{ p \in \cP_l ( k  ) } 2^{ C ( p ) } 
\cr   && = \sum_{ K } x^K Z ( K ) = \sum_{ K } x^K \sum_{ l =1}^{ K } Z_l ( K  )  \,,
\end{eqnarray}
where the definitions of $Z ( K ) $ and $Z_l(K) $ are 
\begin{equation} 
Z ( K )= \sum_{ p \vdash K } 2^{ C ( p ) } \,, \qquad 
Z_l(K) =  \sum_{ p \in \cP_l ( k  ) } 2^{ C ( p ) }  \,.
\end{equation}
Our aim is to recover the first two exponential orders in the large-$K$ asymptotics of $Z(K)$ purely from this combinatorial decomposition.

We will first study the $Z_1(K)$ which receives the contributions from $\mathcal{P}_1(K)$.
Any partition $p \in \mathcal{P}_1(K)$ can be uniquely decomposed into its length-$1$ cycles and a residual partition $r$ containing no length-$1$ cycles:
\begin{eqnarray} 
&& p = [ 1^{ K - \bar m_1 } , r ] \qquad \bar m_1 = \sum_{ i =2}^{ K } i p_i  \in \{ 0 , 1, \cdots , K -1 \} \equiv \cS_1 ( K ) \\ \nonumber 
&& r \vdash \bar m_1 ,\quad  r = [ 2^{ r_2 } , \cdots , K^{ r_K} ] \in \mathcal{P}_2( \bar{m}_1 )\quad \hbox{ with } 
 r_1=0\,,  p_2 = r_2 , \cdots , p_K = r_K   
  \,.
\end{eqnarray}
where $\cP_{ 2^{+} } ( n )$ denotes the partition configuration of $p\vdash n$ excluding the $\mathcal{P}_1$ sets
\begin{equation} 
	\cP_{ 2^{+} } ( n )  = \bigsqcup_{ l \ge 2 } \cP_{ l  } ( n )   \,.
\end{equation}
We have separated out the parts of length $1$ in $p$ and $r $ is the remaining partition of 
$ \bar m_1 $,  which tracks the total number of elements with longer cycles.  

Summing over these decomposed configurations, the exact contribution from the $\mathcal{P}_1(K)$ sector is:
\begin{equation} \label{cZ1}  
Z_1 ( K ) =  \sum_{ \bar m_1 = 0 }^{ K -1 } 2^{ K - \bar m_1 }  \sum_{ r \in \cP_{ 2^{+} } ( \bar m_1 ) } 2^{ C ( r ) }  \,.
\end{equation}
The sum over $r$ precisely generates the sequence of coefficients for a restricted partition function, which can be elegantly extracted using the standard coefficient operator:
\begin{equation}
\sum_{ r \in \cP_{ 2^+} ( \bar m_1 )  } 2^{ C ( r ) }   = \Coeff \left(  x^{ \bar m_1 } , \prod_{n=2}^{ \infty } { 1\over  1 - 2x^n   } \right) 
= \Coeff \left(  x^{ \bar m_1 } , \frac{1}{(2x^2;x)_{\infty}} \right)  \,.
\end{equation}
To extract the leading large-$K$ asymptotics, we recognize that extending the summation limit of $\bar{m}_1$ to infinity introduces only exponentially suppressed errors.  Defining this asymptotic limit as $Z_{1,0}(K)$, the summation transforms directly into the evaluation of the generating function at $x = 1/2$:
\bea\label{ExttoInf} 
Z ( K  ) \sim Z_1 ( K ) \sim Z_{1,0}  ( K )  &&  \equiv  2^{ K } \sum_{ \bar m_1 =0 }^{ \infty } 2^{ - \bar m_1 } \Coeff \left(  x^{ \bar m_1 } , \frac{1}{(2x^2;x)_{\infty}}  \right)   \cr 
&& = 2^K \times \frac{1}{(2x^2;x)_{\infty}} \Bigg\vert_{ x = 1/2} \sim 3.46275 \times 2^K  \,.
\eea
This algebraic derivation exactly reproduces the leading asymptotic coefficient independently obtained via the circles-of-poles method (cf. equations \eqref{eq:pole-expansion} and \eqref{eq:numer-Z2K-suble}).

To extract the next term in the asymptotic expansion, we observe that the leading order approximation 
$Z_1 ( K ) \sim Z_{1,0}  ( K ) $ relied on extending the sum.
While this is strictly valid at the leading exponential order $ 2^K$, it introduces a truncation error. 
We define this remainder term as $ \cR_1 ( K ) $.  
\bea 
Z_1 ( K ) &=& Z_{ 1, 0 } ( K ) - \cR_1 ( K ) \\
\cR_{ 1} ( K ) &=& \sum_{ \bar m_1 = K }^{ \infty  } 2^{ K - \bar m_1 }  \sum_{ r \in \cP_{ 2^{+} } ( \bar m_1 ) } 2^{ C ( r ) }  \,.
\eea 
As we will explicitly demonstrate, this remainder is of exponential order  $ 2^{ K \over 2 } $.
Furthermore, this subleading scaling perfectly matches the leading contribution from the next combinatorial sector
$Z_2(K) \sim Z_{2,0}(K)$, obtained by extending the defining sum to $\infty$, 
Therefore, the first subleading exponential scale intrinsically receives mixed contributions from both the $l=1$ truncation and the $l=2$ sector.

To evaluate $\mathcal{R}_1(K)$, it is convenient to shift the summation variable by defining $m_1' = \bar{m}_1 - K \ge 0$.  
Any partition $r \in \mathcal{P}_{2^+}(K+m_1')$ can be uniquely factored by isolating its length-$2$ cycles. Let $n_1$ denote the number of length-$2$ cycles subtracted from the maximally allowed number $\lfloor (K+m_1')/2 \rfloor$:
\bea  \label{eq:def-subleadingZ2-config}
&&   r \vdash ( K + m_1') \,,  \qquad r = [ 2^{\lfloor  { K+ m_1' \over 2 }    \rfloor   - n_1 } , s ] \cr 
&&   s \in \cP_{ 3^+  } ( 2 n_1  + \bar E ( K + m_1' ) )\,,  \qquad \bar E  ( K + m_1')  = { 1 \over 2 } ( 1 - ( -1)^{ K +  m_1' } ) \cr 
&& C ( r ) = \lfloor  { K+ m_1' \over 2 }    \rfloor  - n_1 + C ( s ) \,,
\eea 
where the integer $0 \le n_1 \le  \Big \lfloor  { K+ m_1' \over 2 }  \Big \rfloor $, and $\bar{E}(n)$ is the parity indicator function, yielding $1$ if $x$ is odd and $0$ if even.
Substituting this combinatorial decomposition back into the remainder gives:
\bea 
&& \cR_1 ( K ) = \sum_{ m_1' =0}^{ \infty }  2^{ - m_1'} \sum_{ n_1 =0}^{ \lfloor { K+ m_1' \over 2 }  \rfloor   }  \sum_{ s \in \cP_{ 3^+  } ( 2 n_1 + \bar E ( K + m_1')  )  } 2^{\lfloor  { K+ m_1' \over 2 }    \rfloor  - n_1 + C ( s )  } \cr 
&& =  \sum_{ m_1' =0}^{ \infty }  2^{\lfloor  { K+ m_1' \over 2 }    \rfloor  - m_1'} \sum_{ n_1 =0}^{ \lfloor { K+ m_1' \over 2 }  \rfloor   }  \sum_{ s \in \cP_{ 3^+  } ( 2 n_1 + \bar E ( K + m_1')  )} 2^{  - n_1 + C ( s )  } \,.
\eea
At large $K$, extending the upper limit of the $n_1$ summation to infinity only discards configurations that are exponentially suppressed compared to the $2^{K/2}$ scale. 
Defining this principal asymptotic component as $\mathcal{R}_{1,0}(K)$, we can get
\bea \label{eq:def-cR10}
 \cR_{ 1 } ( K ) = \cR_{ 1,0} ( K ) + {o}(2^{K/2})   \,.
 \eea 
where 
  \bea 
&&  \cR_{ 1,0} ( K ) = \sum_{ m_1' =0}^{ \infty }  2^{\lfloor  { K+ m_1' \over 2 }    \rfloor  - m_1'} \sum_{ n_1 =0}^{ \infty  }  \sum_{ s \in \cP_{ 3^+  } ( 2 n_1 + \bar E ( K + m_1')  )  } 2^{  - n_1 + C ( s )  }
\cr 
&& =  \sum_{ m_1' =0}^{ \infty }   2^{\lfloor  { K+ m_1' \over 2 }    \rfloor  - m_1'}  
\sum_{ n_1 = 0}^{ \infty } 2^{ -n_1 } \sum_{ s \in \cP_{ 3^+  } ( 2 n_1 + \bar E ( K + m_1')  )    } 2^{ C (s) } \cr  
 && = \sum_{ m_1' =0 }^{ \infty } 
 2^{\lfloor  { K+ m_1' \over 2 }    \rfloor  - m_1'}   \sum_{ n_1 = 0}^{ \infty } 2^{ -n_1 }  
 \Coeff \left (   x^{ 2n_1 + \bar E ( K + m_1')  } , \frac{1}{(2x^3;x)_{\infty}}\right )   \,.
 \eea 
Carrying out the computation detailed in Appendix~\ref{append:higher-order-combinartorics}, we can further rearrange the $\cR_{1,0}$ terms into those with even and odd $K$, respectively:
\begin{eqnarray} \label{simpR10-even} 
&& \cR_{ 1,0} ( 2L )  = 2^{L} \left( 1 + { 1 \over \sqrt{2} } \right)  { 1 \over  ( { 1 \over \sqrt  2 }; { 1 \over \sqrt 2 } )_{\infty} } + 2^L \left( 1 - { 1 \over \sqrt{2} } \right)  { 1 \over ( - { 1 \over \sqrt  2 } ,- { 
  1 \over \sqrt 2 } )_\infty } \,,\\
\label{simpR10-odd}  
&&  \cR_{ 1,0} ( 2L +1 )  = 2^{ L } ( 1 + \sqrt { 2 } ) { 1 \over   ({ 1 \over {\sqrt{2}}} ; {1 \over \sqrt{2} }   )_\infty }  + 2^L ( 1 - \sqrt { 2 } ) {  1 \over  ( -{  1 \over {\sqrt{2}}} ; -{ 1 \over \sqrt{2} } )_\infty }   \,.
\end{eqnarray}

We now generalize this procedure to analyse the contributions from higher $l$ sectors in \eqref{Zsuml}. 
We can write the partitions in 
$ \cP_l( K ) $ as $  [ l^{ p_l } , r ]  $ where $ r$ has all parts of length $(l+1) $ or greater. 
Let $Q(K,l) = \lfloor K/l \rfloor$ denote the maximum possible number of length-$l$ cycles, and $R(K,l) = K \pmod l$ the remainder, such that $K = l Q(K,l) + R(K,l)$. The partition $p$ explicitly takes the form:
\bea 
&& p = [ l^{ p_l } , r ]  ;  \qquad p_l \ge 1; \qquad  r \in \cP_{ l+1} ( \bar m_l ) \\
&& \bar m_l = K - l p_l, \qquad p_l \in \{ 1, 2, \cdots , Q ( K , l ) \} \cr 
&& \bar m_l \in \{ R ( K , l ) , R ( K , l )  + l , \cdots ,  R ( K , l )  + ( Q ( k , l ) -1) l  =( K - l  )   \}  \equiv \cS_l  ( K  ) \nonumber
\eea
Summing over these configurations, the contribution from the $\mathcal{P}_l(K)$ sector evaluates to:
\bea 
Z_l  ( K  ) =  \sum_{ \bar m_l  \in \cS ( K , l ) } 2^{ ( K - \bar m_l )\over l } \sum_{ r \in  \cP_{ (l+1)^{+} }  ( \bar m_l   ) }    2^{ C ( r ) } 
\eea
where, following our previous notation, the set $\mathcal{P}{ (l+1)^{+} } ( n ) \equiv \bigsqcup_{j \ge l+1} \mathcal{P}_j(n)$ denotes the partitions of $n$ completely devoid of any cycles of length $l$ or smaller.
The leading contributions come from small $ \bar m_l$. E.g. when $ K $ is divisible by $l$ the smallest $ \bar m_l = 0 $ and we have $2^{ K / l } $ which agrees with the form of the answer from the partial fractions method.
We also observe that 
\bea 
\sum_{ r \in  \cP_{ l+1}  ( \bar m_l  ) }    2^{ C ( r ) }  
= \Coeff \left( x^{ \bar m_l  } , \frac{1}{(2x^{l+1};x)_{\infty}} \right)  \,.
\eea  
In this framework, the $l=1$ sector dictates the leading asymptotic growth, the $l=2$ sector provides the first exponentially suppressed non-perturbative correction, $l=3$ the second, and so forth. 
Within any given sector $l$, the dominant contributions strictly arise from the smallest possible values of $\bar{m}_l$. For instance, if $K$ is completely divisible by $l$, the minimum is $\bar{m}_l = 0$, yielding a leading weight of $2^{K/l}$. 
To formalize the large-$K$ limit for a fixed small $l$, let us assume for simplicity that $K$ is a multiple of $l$. 
The set of remainder $\bar{m}_l$ simplifies to $\mathcal{S}_l ( K ) = \{ 0 , l , 2l , \dots \}$. Writing $\bar{m}_l = l \hat{m}_l$ with integer $\hat{m}_l \ge 0$, and extending the summation to infinity (discarding deeply suppressed errors), we define the leading asymptotic component $Z_{l,0}(K)$:
\bea\label{lapprox}
Z_l ( K ) \sim  Z_{ l, 0 } ( K )   \equiv 2^{ K /l } \sum_{ \hat m_l = 0   }^{ \infty } 2^{ - \hat m_l }  
\Coeff \left( x^{  l \hat  m_l    } ,  \frac{1}{(2x^{l+1};x)_{\infty}}  \right)  \,.
\eea
Specialisation to the case $ l=2$ results in the $Z_{2,0}(K)$ for even $ K = 2L$, 
\bea\label{simpZ10}  
Z_{ 2 , 0 } ( 2L ) &&  = 2^{ L }  \cdot { 1 \over 2 } \left ( \frac{1}{(2x^3;x)_{\infty}} \Big|_{x=\frac{1}{\sqrt{2}}} 
+\frac{1}{(2x^3;x)_{\infty}} \Big|_{x=-\frac{1}{\sqrt{2}}}  \right )  \cr 
&& = 2^{ L} \cdot { 1 \over 2 }  \left ({ 1 \over  ( { 1 \over \sqrt  2 } ; { 1 \over \sqrt 2 } )_\infty }  +  { 1 \over  ( -  {1 \over  \sqrt 2 } ; -  { 1 \over \sqrt 2 }   )_\infty }    \right )  \,.
\eea 
and odd $K = 2L+1$ respectively as
\bea 
Z_{ 2 , 0 } ( 2L +1 ) = { 2^L \over \sqrt { 2 } }  \left ({ 1 \over  ( { 1 \over \sqrt  2 } , { 1 \over \sqrt 2 } )_{\infty}}  +  { 1 \over ( -  {1 \over  \sqrt 2 }; -  { 1 \over \sqrt 2 } )_\infty }    \right )  \,.
\eea

Collecting terms we have, for even $K$,  using \eqref{simpR10-even} and \eqref{simpZ10}
\bea 
&& Z_{ 2 , 0 } ( 2L  )  - \cR_{ 1, 0 } ( 2L  ) \\
&=& - {  2^{ K \over 2 } \over 2 }   \left ({ 1 \over  ( { 1 \over \sqrt  2 } ; { 1 \over \sqrt 2 })_\infty }  +  { 1 \over  ( - { 1 \over  \sqrt 2 } ; -  { 1 \over \sqrt 2 })_\infty }    \right )  - {  2^{ K \over 2 } \sqrt {2} \over 2   }    \left ({ 1 \over  ( { 1 \over \sqrt  2 } , { 1 \over  \sqrt 2 } )_{\infty} }  -   { 1 \over  ( - { 1 \over  \sqrt 2 } , -  { 1 \over \sqrt 2 } )_\infty }    \right ) \cr 
&& = - {  2^{ K /2} \over 2 } ( 1 + { \sqrt {2}  } ) { 1 \over ( { 1 \over \sqrt  2 } ;{ 1 \over \sqrt 2 } )_\infty } - { 2^{ K /2} \over 2 } (  1 - \sqrt{2}  ) 
 { 1 \over  ( - { 1 \over  \sqrt 2 } ; -  { 1 \over \sqrt 2 } )_\infty } \,,
\eea 
and for odd $K$ 
\bea 
Z_{ 2 , 0 } ( 2L +1  )  - \cR_{ 1, 0 } ( 2L +1  )  
= { - 2^L \over \sqrt{2} } \left (   { ( 1 + \sqrt { 2 } )  \over  ( { 1 \over \sqrt  2 } ; { 1 \over \sqrt 2 } )_\infty }    +    { (  \sqrt { 2 } -1 )  \over  ( - { 1 \over \sqrt  2 } ; - { 1 \over \sqrt 2 } )_\infty }     \right )  \,.
\eea 
These expressions agree with the leading and subleading formula obtained from the complex-analytic method, i.e. the $n=2$ term in \eqref{eq:degeneracyZ2}, with the explicit form elaborated in \eqref{fstfew} and numerical values in \eqref{eq:numer-Z2K-suble}. 
This provides a combinatorial explanation to the degeneracy counting function derived by complex analysis method in Section \ref{sec:complexnalysis}. 
Complex analysis independently shows that there is  an optimal truncation of the asymptotic series at the order $\sqrt{K}$, as  in \eqref{eq:saddlepoint-Kn}. It is an interesting open problem to derive this optimal truncation from the combinatorial framework. While the disjoint decomposition \eqref{PartitionsDisjoint} is exact at  finite $K$, the extensions of   upper summation limits (e.g. of  $\bar{m}_1$ from \eqref{cZ1}  to \eqref{ExttoInf})   to infinity in the derivation of the first two terms of the asymptotic expansion above will require revision for sufficiently high orders compared to $K$. Elucidating the detailed mechanisms at play in this interaction between combinatorics and complex analysis should be an illuminating project.

\subsubsection{Refined partition function}
\label{sec:small-cycles-refined}

The small-cycle organisation of the unrefined partition function admits a natural extension to the 
 refined partition function \eqref{eq:refined-coefficnet-partition}:
\begin{align}
	\begin{split}
Z(m,n)&=\sum_{ p \vdash m , q \vdash n } \prod_{i}    { ( p_i + q_i ) ! \over p_i! q_i ! } = \sum_{l=1}^{\infty}Z_l (m,n)\\
Z_l (m,n)&= \sum_{ ( p \circ q ) \in \mathcal{P}_l(m+n)}  { ( p_i + q_i ) ! \over p_i! q_i ! }  \,.
	\end{split}
\end{align}
The sum over the partitions $ p \vdash m , q \vdash n $ is organised by the positive integer $l$, by using the constraint that the partition $ ( p \circ q ) \vdash (m+n)$ has a minimum part (cycle length of $ \gamma$ )  of size $l$. 
Using this combinatorial method, we can systematically derive the subleading perturbative corrections to the leading-order asymptotics established in the limit $m,n\to \infty$ and $m/n$ fixed \cite{Ramgoolam:2018epz}. 
The mathematical connection between the analytic geometry of the asymptotic expansion and the combinatorics of small cycle dominance allows us to reliably extend the all-orders asymptotics to the refined counting.
The terms are organized such that contributions from $\mathcal{P}_l$ are exponentially suppressed compared to $Z_1(m,n)$. 
A key operational advantage of this combinatorial approach is that extracting subleading corrections is considerably more straightforward than executing rigorous multi-variable complex analysis. 
Conversely, this combinatorial mapping provides a concrete guide for decoding the exact multi-variate analytic structure in future studies.

To isolate the leading perturbative corrections, the relevant contributions arise exclusively from $Z_1(m,n)$.
The corresponding $\mathcal{P}_1$ configurations are explicitly parametrized as:
\begin{equation}
\mathcal{P}_1(m) =[1^{m-\bar{m}_1},r], \quad r\vdash \bar{m}_1; \qquad  \mathcal{P}_1(n) =[1^{n-\bar{n}_1},s], \quad s\vdash \bar{n}_1 \,.
\end{equation}
The associated degeneracy counting function is then given by:
\begin{align}\label{eq:first-order-refined}
	\begin{split}
Z_1(m,n) = \sum_{\bar{m}_1=0}^{m-1} \sum_{\bar{n}_1=0}^{n-1} \frac{(m+n-\bar{m}_1 - \bar{n}_1)!}{(m-\bar{m}_1)!(n-\bar{n}_1)!}
		\text{Coeff} \left[
		x^{\bar{m}_1} y^{\bar{n}_1},\, \prod_{i=2}^\infty \frac{1}{1-x^i -y^i}
		\right] \,.
	\end{split}
\end{align}
The leading order of asymptotics computed in \cite{Ramgoolam:2018epz} can be rederived by applying the Stirling formula, and keep $m/n$ fixed in the $m,n\to \infty$ limit, yielding
\begin{equation}
	\frac{(m+n-\bar{m}_1 - \bar{n}_1)!}{(m-\bar{m}_1)!(n-\bar{n}_1)!} \sim  \frac{1}{\sqrt{2\pi }} \frac{(m+n)^{m+n}}{m^m n^n} \sqrt{\frac{m+n}{mn}} 	\left( \frac{m}{m+n}\right)^{\bar{m}_1} \left( \frac{n}{m+n}\right)^{\bar{n}_1}  \,.
\end{equation}
Denote $\lambda = \frac{m}{m+n}$,
the coefficients in \eqref{eq:first-order-refined} can then be resumed to obtain:
\begin{equation}\label{eq:refined-leading-mn}
Z_1(m,n)	\sim \frac{1}{\sqrt{2\pi }} \frac{(m+n)^{m+n}}{m^m n^n} \sqrt{\frac{m+n}{mn}} \mathcal{Z}_2(\lambda,1-\lambda) \,,
\end{equation}
where $\mathcal{Z}_2(x,y)$ is defined in  \eqref{eq:Zlxy-function} and 
\eqref{eq:refined-leading-mn} is completely consistent with \cite{Ramgoolam:2018epz}. It is useful to rewrite this in terms of $ K = (m+n) $ and $ \lambda = \frac{m}{K}$. 
\begin{equation}
Z_1( K \lambda, K(1- \lambda) ) \sim \frac{1}{\sqrt{2\pi K  }} \left(\frac{1}{\lambda}\right)^{K \lambda  + \frac{1}{2} }  \left(\frac{1}{ 1-\lambda}\right )^{(1- \lambda) K + \frac{1}{2} }  \mathcal{Z}_2(\lambda,1-\lambda) \,,
\end{equation}

While the leading-order behaviour \eqref{eq:refined-leading-mn} was originally derived by evaluating the singular curve $1-x-y=0$ \cite{Ramgoolam:2018epz}, our combinatorial formulation provides a direct mechanism to extract the subleading perturbative corrections. 
These corrections originate from the higher-order terms in Stirling's approximation, namely:
\begin{align}
\begin{split}
&	m! \sim \frac{m^{m+\frac{1}{2}}}{e^m} \sqrt{2\pi} \left( 1+\frac{1}{12m}\right) \,, \\
& (m-\bar{m}_1)^{m-\bar{m}_1}\sim m^{m-\bar{m}_1} e^{-\bar{m}_1} \exp\left(\frac{\bar{m}_1^2}{2m} \right) \,.
	\end{split}
\end{align}
Substituting these expansions into the summand of \eqref{eq:first-order-refined} yields the subleading corrections for the prefactor:
\begin{eqnarray} \label{eq:subleading-corrections-refined}
&& \frac{(m+n-\bar{m}_1 -\bar{n}_1)!}{(m-\bar{m}_1 )! (n-\bar{n}_1 )! }  \sim \frac{(m+n)^{m+n+\frac{1}{2}}}{ \sqrt{2\pi}m^{m+\frac{1}{2}} n^{n+\frac{1}{2}}}
\left(\frac{m}{m+n}\right)^{\bar{m}_1}  \left(\frac{n}{m+n}\right)^{\bar{n}_1} \\
&& \times \left[
1- \frac{m+n}{2mn} \left(\frac{\lambda^2-\lambda+1}{6} - \bar{n}_1 \lambda^2 -(1-\lambda)^2 \bar{m}_1 + (\lambda \bar{n}_1 - (1-\lambda) \bar{m}_1)^2
\right)
\right] \nonumber \,.
\end{eqnarray} 
We parameterize these subleading corrections using the symmetric ratio $\frac{m+n}{mn}$ in the homogeneous thermodynamic limit $m,n\to\infty$ and $m/n$ fixed.
Performing the summation over $\bar{m}_1$ and $\bar{n}_1$ effectively translates the polynomial insertions into differential operators acting on the generating function, resulting in the asymptotic formula for $Z_1(m,n)$, with a first sub-leading correction:
\begin{eqnarray}
Z_1(m,n) && \sim \frac{1}{\sqrt{2\pi }} \frac{(m+n)^{m+n}}{m^m n^n} \sqrt{\frac{m+n}{mn}}  \left[1-\frac{m+n}{2mn} \hat{O}\right] \mathcal{Z}_2(x,y) \Big|_{(x,y)=(\lambda,1-\lambda)} 
	\end{eqnarray}
where the operator $\hat{O}$ is explicitly
\begin{eqnarray}
\hat{O}& =& \frac{x^2-xy +y^2}{6} -xy (x \partial_y+y \partial_x) + x^2(y\partial_y)^2 -2(x\partial_x)(y\partial_y) +y^2 (x\partial_x)^2  \,.
	\end{eqnarray}
It is again useful to rewrite as 
\begin{eqnarray}
	&& 	Z_1(K\lambda,K(1-\lambda)) \\ &\sim& \frac{1}{\sqrt{2\pi K  }} \left(\frac{1}{\lambda}\right)^{K \lambda  + \frac{1}{2} }  \left(\frac{1}{ 1-\lambda}\right )^{(1- \lambda) K + \frac{1}{2} } 
		 \left[1-\frac{\hat{O}}{2K \lambda ( 1 - \lambda )} \right] \mathcal{Z}_2(x,y) \Big|_{(x,y)=(\lambda,1-\lambda)}  \nonumber \,.
\end{eqnarray}
	
The derivation of the first sub-leading correction for $Z_1(m,n)$ above  demonstrates that the combinatorial mechanism of small-cycle dominance remains a robust organizing principle in the refined sector. By parameterizing the $\mathcal{P}_1$ configurations of the equivalence-generating permutations $\gamma$, we have shown how the discrete counting of invariants reduces to specific trace structures in $\sigma$ through the action of differential operators on the generating function. However, a complete analytic description of the all-orders asymptotics for these multi-variable functions remains an open geometric challenge.
	
While the univariate partition function is characterized by isolated simple poles accumulating toward a natural boundary, its multivariate refinement introduces a richer analytic geometry. 
The singular locus is no longer composed of discrete points, but rather complex hypersurfaces defined by equations of the form $1 - \sum x_j^i = 0$. 
These singular hypersurfaces can intersect nontrivially strictly inside the open unit sphere. 
For instance, the simultaneous vanishing of distinct denominator factors, such as
\begin{equation}
	1- x^3-y^3=0, \qquad 1-x^4-y^4=0 \,,
\end{equation}
yields intersection loci at approximately $(x,y) = (-0.2020 \pm 0.9158 i,  -0.2020 \mp 0.9158i)$.
At these intersecting points, the generating function develops higher-order singularities (poles of order two or greater).
The emergence of such intersecting singular strata significantly complicates the complex-analytic extraction of asymptotic coefficients.
Future work can aim for a global geometric generalization of small-cycle dominance, and its link to the contributions of these degenerate poles, likely utilizing the framework of Analytic Combinatorics in Several Variables (ACSV) as detailed in  \cite{PWM2024}. 

An alternative limit, investigated in \cite{Dolan:2007rq}, involves keeping $m$ fixed while exploring the asymptotics of the degeneracy in the large-$n$ limit. 
In this regime, the dynamics are dominated by the partitions of $n$, which grow according to the Hardy-Ramanujan formula. 
Specifically, for a fixed partition $p=[1^{p_1},2^{p_2}\dots] \vdash m$, the corresponding coefficient grows as:
\begin{equation}\label{eq:Dolan-counting-2-matrix}
Z_1(p ,n) \sim \frac{1}{4\sqrt{3}n} \left(\frac{\sqrt{6} n}{\pi} \right)^{C(p)} \frac{1}{\text{Sym}( p )} e^{\sqrt{\frac{2n}{3}} \pi} \,.
\end{equation}
An interesting open question is how to reproduce \eqref{eq:Dolan-counting-2-matrix} directly from the refined formula \eqref{eq:first-order-refined} by summing over the partitions of $n$ for a fixed partition $p$.

\section{Circles of poles  and asymptotics of weighted partition numbers   }
\label{sec:CircSings} 

We obtained in  Section \ref{sec:complexnalysis} the  all-orders asymptotics for the unrefined multi-matrix partition function $ \cZ_d ( x ) $, by using complex analysis on the $x$-plane. 
Each order parametrised by $n\in \{ 1, 2, \cdots \} $   corresponds to simple poles at radii $ |x_n| = d^{ -1/n}  = e^{ -{ \ln d  \over n}}$, with increasing  orders at higher $n$   suppressed by exponential functions of $ \ln d $. As $ n $ tends to infinity,  the radii $ |x_n| $ converge to $1$, the circle of radius $1$ being a natural boundary where there is the accumulation of an infinite number of poles. As explained in section \ref{sec:MTconeq}, 
 these increasing radii are associated with contributions from increasing minimal-cycle lengths of {\it equivalence permutations} which correspond to different factors in the infinite product formula for $ \cZ_d ( x ) $. 
 
In this section, we explain using known results from analytic combinatorics (as collected notably in \cite{FlajoletSedgewick:2009}) that this mechanism for all-orders asymptotics extends to any generating function $ \cZ ( x ) $ which is meromorphic within the unit disc and has a natural boundary at $ |x| =1$.
This  analytic structure is closely related to the Pólya-Carlson theorem \cite{FlajoletSedgewick:2009} characterising a large class of combinatorial generating functions. 
The theorem asserts a dichotomy: any power series $f(z) = \sum_{n} a_n z^n$ with strictly integer coefficients and a radius of convergence $R=1$ must either be a simple rational function (i.e., a ratio of two polynomials) or it must admit the unit circle as a natural boundary.
Motivated by this generic feature, our methodology applies to a such meromorphic  generating functions developing a natural boundary on the unit circle.

As explicit examples within this more general context, we consider the generating functions for weighted partitions (see \eqref{Zweighted}), where the constant $n$-independent weights $d$ of $ \cZ_d(x)$ are replaced by more general weights $w_n$. 
We work out the first two leading asymptotic terms for two examples of weighted partition counting. The first example has the simplest form of small cycle dominance of the same form as $ \cZ_d ( x ) $.
 
More generally, we have increasing minimal cycle lengths being sub-dominant for minimal cycle lengths above some $n_0$, but there is a non-trivial re-arrangement between orders of asymptotics and minimal cycle lengths below $ n_0$. This is illustrated in the second example. The non-trivial re-arrangement can be understood from the correspondence between  cycle lengths and powers of $x$ in the factors of  generating function for weighted partitions, and corresponds on the complex analysis side, to a re-ordering between the parameter $n$ in the infinite product and the ordering of the circles of singularities in terms of increasing distance from the origin.

\subsection{All-orders asymptotics for meromorphic functions}
\label{ssec:all-order-meromo} 

We will consider functions $ F ( z ) $ which are analytic around the origin and have a sequence of poles lying on circles of radii $r_1 < r_2 < ... $. The radii accumulate at $R =1$. There are a finite number of poles at each radius. 
This general analytic structure naturally accommodates both the unrefined $d$-matrix model and the weighted partition cases which will be discussed in this section.
 We will give adapted versions of Theorem IV.7 and IV.10 of \cite{FlajoletSedgewick:2009} which are suitable  for this class of functions. 

\noindent 
{\bf Lemma 5.1 (Theorem IV.7 of \cite{FlajoletSedgewick:2009})}  If  $F(z)$  is analytic at zero and $R$ is the modulus of the nearest singularity to the origin, then the coefficient of $z^n$ in the convergent expansion of $F(z) $, denoted $f_n$,  satisfies 
\bea 
f_n  \simexp  R^{ -n } \,, 
\eea
which is read as $f_n$ has \emph{exponential type} (short for (e.t.)) $ R^{ -n} $.  Technically, this dictates that the asymptotic growth rate of the coefficients is strictly bounded by any arbitrarily small deviation from $R^{-1}$
\bea 
&& \lim_{ n \rightarrow \infty } { f_n \over A^n }  = \infty  ~~~ \hbox{if} ~~  A < R^{-1} \cr 
&& \lim_{ n \rightarrow \infty } { f_n \over A^n }  = 0 ~~~ \hbox{if} ~~  A > R^{-1} \,.
\eea

In our specific setup, we assume that $F(z)$ is meromorphic within the disk $ |z| \le  r < 1 $. 
Let the closest poles be at a distance $ |z| = r_1$. 
The generating function admits a convergent expansion
$ 
\sum_{ K =0}^{ \infty } f_K z^K 
$,
whose strict radius of convergence is exactly $ r_1  $. 
Directly applying Lemma 5.1, we conclude that the large-$K$ asymptotic behaviour of the sequence $f_K$ is governed by this dominant singularity scale:
\bea 
f_K \simexp r_1^{ -K }   \,.
 \eea 

Let  $F_1(z) $ denote the total polar part contributed by the innermost circle of singularities at $|z|=r_1$. Summing over all distinct poles $z_{1,i}$ located at this distance from the origin (labeled by $i$), the  polar part is given by:
\bea 
F_1 ( z ) = \sum_{ {  | z_{1;i} | = r_1  }   }   \cP ( F ( z ) , z_{ 1 ; i} )   \,,
\eea
where $\mathcal{P} ( F(z) , z_{1,i} )$ denotes the principal part of the Laurent expansion of $F(z)$ centered at the pole $z_{1,i}$. By subtracting this leading polar part, the residual function $F(z) - F_1(z)$ effectively has its nearest singularities removed, granting it a convergent expansion strictly within a larger disk of radius $r_2$, which is the distance to the next set of poles. 
Applying Lemma 5.1 to this subtracted function, the asymptotic growth of its Taylor coefficients is bounded by the new radius of convergence:
\bea
[ z^K ] ( F( z) - F_1 ( z ) ) \simexp r_2^{ - K } \,.
\eea 

This subtraction procedure can be further generalized. At the $m$-th concentric circle of poles $|z| = r_m$, the corresponding polar part is defined as:
\bea 
F_m  ( z ) = \sum_{ | z_{ m ; i } | = r_m    } \cP ( F ( z ) , z_{ m ; i } )  \,.
\eea
By removing the cumulative contributions of all poles up to the $m$-th circle, the remainder function
\bea 
F (z) - \sum_{ q =1}^{ m } F_q ( z ) 
\eea
becomes completely analytic within the disk $|z| < r_{m+1}$, where $r_{m+1}$ marks the distance to the $(m+1)$-th layer of singularities.
Applying Lemma 5.1 for each $m$, we have 
\bea 
[ z^K ]  ( F (z) - \sum_{ q =1}^{ m } F_q ( z )  )   \simexp r_{ m+1}^{-K} \,.
\eea

Furthermore, each isolated composite polar function $F_q(z)$ admits its own formal Taylor expansion around the origin:
\bea 
F_q ( z ) = \sum_{ K =0}^{ \infty } z^K f_{ q;K }  \,.
\eea
Since the singularities of $F_q(z)$ are situated exactly at $|z| = r_q$, a direct application of Lemma 5.1 confirms that its coefficients independently obey the expected exponential scaling: 
\bea 
f_{ q;K } \simexp r_q^{ - K }  \,.
\eea
This leads to the following proposition:
\begin{proposition}\label{propWkAs}
The sequence of functions $ f_{ q;K } $ forms a weak-asymptotic expansion for $ f_K$. 
\end{proposition}
From the defintion of weak asymptotic expansions ( e.g. \cite{wiki:asymptotic_expansion,olver,benderorszag} )  this equivalently means that 
\bea\label{DefWkAs}  
\lim_{ K \rightarrow \infty }   {  f_K  - \sum_{ q =1}^{ m } f_{ q;K }    \over  f_{ m;K } }   = 0  \,.
\eea

\begin{proof}
Using Lemma 5.1, 
\bea 
f_K  - \sum_{ q =1}^{ m } f_{ q;K }    \simexp r_{ m+1}^{ - K } \hbox{ for large $K$ }  \,.
\eea
On the other hand, 
\bea 
 f_{ m;K } \simexp r_{ m }^{ - K }   \hbox{ for  large $K$ }  \,.
\eea
The ratio in \eqref{DefWkAs} behaves as $ (r_{ m+1} /r_m )^{ -K } $ which vanishes at large $K$, since $ r_{ m+1} /r_m > 1$.  The proves the proposition. 
\end{proof}

\noindent 
{\bf Lemma 5.2 (Theorem IV.10 of \cite{FlajoletSedgewick:2009})} A function $F(z)$ which is analytic at $ z= 0 $, at $ |z|=R$ and has a finite number of poles $ \alpha_1 , \cdots , \alpha_m$ of orders $ s_1 , \cdots , s_m $  for $ |z| < R$, has coefficients of its Taylor expansion around $ z=0$,  behaving as
\bea\label{Lem5pt2}  
f_K \sim \sum_{ i =1}^m \Pi_j ( K ) \alpha_j^{ -K } + O ( R^{ - K  } ) \,,
\eea
where each $\Pi_j ( K )$ is a polynomial in $K$ of degree exactly $ (s_j -1)$. 

We will now  prove the following proposition. 
\begin{proposition} 
	The sequence of functions $f_{ q;K } $ forms a strong-asymptotic expansion for $f_{ K } $, i.e. 
	\bea\label{defstrng} 
	\lim_{ K \rightarrow \infty }   {   f_{ K }  - \sum_{ q =1}^{ m } f_{ q;K }    \over  f_{ m+1;K }  }   = 1 \,,
	\eea
\end{proposition}
The proof of this proposition will be in Appendix \ref{appen:proof}.

\subsection{Weighted partitions } 
\label{sec:Wtparts} 

Our complex-analytic methodology can be readily applied to a variety of combinatorial models to extract exact degeneracy counting functions. 
In the rest of the section, we will consider another class of generalized weighted partition functions of the form by the similar algorithm:
\begin{equation}\label{Zweighted}  
	\mathcal{Z}( \vec w ; x) = \prod_{n=1}^\infty \frac{1}{1 - w_n x^n} = \sum_{K=0}^\infty Z( \vec w ; K) x^K \,. 
\end{equation}
Physically, $Z(\vec w ; K)$ counts partitions of $K$ where each part $n$ is weighted by $w_n$. From an analytic perspective, the structure of $Z(\vec w ;K)$ is governed by the radii of the poles $r_n = w_n^{-1/n}$. 
To ensure a natural boundary at the unit circle $|x|=1$, we require the weights to be sub-exponential, i.e., $\lim_{n \to \infty} \frac{1}{n} \ln |w_n| = 0$. This condition is satisfied by weights that grow polynomially (e.g., $w_n \sim n^a$) or as $w_n \sim e^{n^\alpha}$ with $\alpha < 1$. When the specific weights $\vec{w}$ are evident from the context, we will adopt the simplified notation $ \mathcal{Z}( \vec w ; x) \rightarrow \cZ ( x )  $ and $ Z ( \vec w ; K ) \rightarrow Z ( K ) $.

\subsubsection*{Example 1: Partition with increased weights $w_n=n+1$} 
\label{sec:Wtprts1}

Consider the generating function
\begin{equation} \label{eq:n+1partitoinfunction}
\mathcal{Z}(x) = \prod_{ n=1}^{ \infty  } { 1 \over  1 - ( n+1) x^n  }  = \sum_{K=0}^\infty \tilde{Z}(K) x^K \,,
\end{equation}
which serves as the generating function of the series OEIS A074141 \cite{oeisA074141}.
The first few terms of this sequence are:
\begin{equation}
	1, 2, 7, 18, 50, 118, 301, 684, 1621, 3620, 8193, \cdots
\end{equation}
This sequence represents the sum over partitions of $K$, where each partition is weighted by the plus one shifts of all parts of the partition. A crucial feature of this model is the strict monotonicity of the pole radii $r_n = (n+1)^{-\frac{1}{n}}$. 
This property is equivalent to the inequality, for all $n \in \{ 1, 2, \cdots \}$
\begin{equation} \label{eq:}
	(n+2)^n < (n+1)^{n+1} \iff (n+1) >\left( 1 + \frac{1}{n+1} \right)^n \,. 
\end{equation}
This follows since the binomial expansion gives 
\bea 
n+1  - \left( 1 + \frac{1}{n+1} \right)^n  = n - \sum_{j=1}^{ n }  { 1 \over j! }  { n ( n-1) \cdots ( n - j +1 ) \over (n+1)^j } >0  \,.
\eea
The sum has $n$ terms each of which is less than $1$. 

Because $r_n$ monotonically increases toward 1, the principle of ``Small Cycle Dominance" holds strictly: the singularity closest to the origin is the simple pole at $x_{1;0} = 1/2$. 
Furthermore, all singularities are non-degenerate, forming an exact layer structure of isolated concentric poles. 
We can thus apply the partial fraction decomposition from \eqref{eq:pole-expansion}. 
By expanding the partition function asymptotically in terms of its poles and residues, we obtain the coefficients:
\begin{equation}
\tilde{Z}(K) \sim \sum_{n=1}^\infty \sum_{j=0}^{n-1} c_{n;j} \omega_{n}^{-j K} (n+1)^{\frac{K}{n}} \,.
\end{equation}
This structure closely mirrors our analysis of the constant-weight $d$-matrix model. 
Consequently, the leading asymptotic behaviour of $Z(K)$ is governed by the $n=1$ pole while the first subleading correction is dictated by the $n=2$ roots.
Truncating the asymptotic series to the first two orders yields:
\begin{equation}
\tilde{Z}(K) \sim 18.56\times 2^K -114.48\times 3^{\frac{K}{2}} +0.266\times (-1)^K \times 3^{\frac{K}{2}} +\cdots \,,
\end{equation}
where the prefactor of the leading term is explicitly evaluated via the infinite product:
\begin{equation} 
\prod_{ k=2}^{ \infty } { 1 \over  1 - ( k+1) 2^{ - k }  }  \approx 18.56 \,.
\end{equation} 
 
To better understand the underlying configurations, we can geometrically expand each factor in the product \eqref{eq:n+1partitoinfunction} as follows:
\begin{equation}
\mathcal{Z}(x) = \prod_{n=1}^\infty \sum_{p_n=0}^\infty(n+1)^{p_n}x^{n p_n} = \sum_{K=0}^\infty  \sum_{p\vdash K}  \prod_{n=1}^K (n+1)^{p_n} x^K\,.
\end{equation}
The coefficients of this series can then be systematically collected to yield :
\begin{equation}
	\tilde{Z}(K) = \sum_{p \vdash K} \prod_{i=1}^K (i+1)^{p_i} \equiv \sum_{p\vdash K} \tilde{N}(p) \,.
\end{equation}

We can then classify the sum over all partitions in terms of the subsets $\mathcal{P}_l(K)$. 
 Specifically, the contribution from $\mathcal{P}_1(K)$, which encapsulates all partitions where the smallest part has a length of $1$ (parametrized as $[1^{K-m_1}, r]$ with $r \vdash m_1$), can be evaluated as:
\begin{align}
	\begin{split}
\tilde{Z}_1 (K) &= \sum_{m_1=0}^{K-1} 2^{K-m_1} \prod_{i=2}^K (i+1)^{p_i} \\
& \sim 2^K \sum_{m_1=0}^\infty 2^{-m_1} \Coeff \left(x^{m_1},\prod_{n=2}^\infty \frac{1}{1-(n+1)x^n}\right) \\
&= 2^K \times \prod_{n=2}^\infty \frac{1}{1-(n+1)2^{-n}} \,.
	\end{split}
\end{align}
This analytic result exactly reproduces the leading asymptotic prefactor, confirming that the $\mathcal{P}_1$ configurations—those dominated by the smallest cycle length—are indeed the dominant contributors to the overall degeneracy.

Notably, the first non-perturbative correction remains negative, a feature consistent with the behaviour observed in the unrefined $d$-matrix theory.
This negative sign is not a mere numerical artifact but a direct mathematical consequence of the small cycle dominance mechanism.
Combinatorially, the leading asymptotic behaviour isolates partitions dominated by the set $\mathcal{P}_1(K)=[1^{K-m_1},r]$, $r\vdash m_1$, 
with smallest cycle of length $1$.
As the $2$-matrix model discussed in Section \ref{ssec:unrefined-smallcycle}, to obtain a closed-form analytic expression for the leading order, we relax this constraint and extend the summation over $m_1$ to infinity.
Consequently, the computation of the first subleading correction must not only incorporate the genuine contributions from the next dominant cycle structures (such as partitions dominated by $\mathcal{P}_2(K)$),
but also subtract the over-counted artificial "tail" of the 1-cycle summation from $m_1 = K$ to $\infty$. 
This subtraction process, which inherently mirrors the procedure in $d$-matrix theory, naturally generates the negative prefactor. 
Therefore, the negative sign is inextricably linked to the algebraic structure of the small cycle dominance approximation itself.

\subsubsection*{Example 2: $w_n=n$, partition norms} 
\label{sec:Wgtparts2}

There are also examples where the smallest cycle does not correspond to the dominant configuration. A specific weight choice of $w_n=n$ yields the partition function 
\begin{equation} 
	\mathcal{Z} ( x ) = \prod_{ n=1}^{ \infty } { 1 \over 1 - n x^n } =\sum_{K=0}^\infty \dot{Z}(K) x^K \,.
\end{equation}
whose coefficients $\dot{Z}(K)$,  the first few of which correspond to the sequence OEIS A006906 \cite{oeisA006906}, are:
\begin{equation}
	1, 1, 3, 6, 14, 25, 56, 97, 198, 354, 672  \cdots
\end{equation}
The coefficient counts the weighted summation over all the partitions of $K$
\begin{equation}\label{eq:def-Np}
	\dot{Z}(K) = \sum_{p\vdash K} N(p)\,, \quad	N(p) = \prod_{i=1}^K i^{p_i} \,, \quad p=[1^{p_1},2^{p_2}\cdots K^{p_K}]\,,
\end{equation}
where $N(p)$ is known as the \emph{norm of partition} defined for given partition $p\vdash K$ \cite{sills2019product,kumar2021analytic,Rana2025partition},
acting as a statistical weight for each partition of the integer $K$. 
Combinatorially, the function $\dot{Z}(K)$ counts the number of \emph{dotted Young diagrams}, where exactly one dot is placed in each row of the Young diagram.

The singularities are located at $x^n = n^{-1}$, corresponding to pole radii  $r_n = n^{-\frac{1}{n}}$. 
Crucially, this sequence of radii is not monotonic: it starts at $r_1=1$, reaches a global minimum at $n=3$  ($r_3 = 3^{-1/3} \approx 0.693$), and then monotonically increases back towards $1$ as $n \to \infty$.
Furthermore, the identity $r_2=r_4$ implies that the corresponding poles next to the origin are degenerate.
Applying the partial fraction decomposition, and carefully accounting for these second-order degenerate poles, the leading terms—ordered by their distance from the origin—are given by:
\begin{align}
	\begin{split}
		\mathcal{Z}(x) &= \frac{c_{3;0}}{1-3^{\frac{1}{3}}x} + \frac{c_{3;1}}{1-3^{\frac{1}{3}} \omega_3^{-1} x} + \frac{c_{3;2}}{1-3^{\frac{1}{3}} \omega_3^{-2} x}  \\
		&+ \frac{d_1}{(1-\sqrt{2}x)^2} + \frac{d_2}{(1+\sqrt{2}x)^2}
		+ \frac{d_3}{1-\sqrt{2}x} + \frac{d_4}{1+\sqrt{2}x} \\
		& + \frac{c_{4;0}}{1+i\sqrt{2}  x} + \frac{c_{4;1}}{1-i\sqrt{2}  x} +\cdots \,,
	\end{split} 
\end{align}
where $d_{1,2,3,4}$ are terms contributed by the degenerate poles from $n=2,4$ factors.
These residue coefficients can be numerically evaluated as in \eqref{eq:def-cnj}.
The Taylor series expansion of the right-hand side then yields the asymptotic behaviour:
\begin{align}\label{eq:ZKwn=n}
	\begin{split}
		Z(K) &= 97923\times 3^{\frac{K}{3}} + (0.114+0.01i)\times 3^{\frac{K}{3}} e^{-\frac{2\pi i}{3} K} + (0.114-0.01i)\times 3^{\frac{K}{3}} e^{\frac{2\pi i}{3} K} \\
		&- [9165.9(K+1) - 460275] \times 2^{\frac{K}{2}} + [0.091(K+1)-1.747] \times (-1)^K 2^{\frac{K}{2}} \\
		& +(0.039+0.008i)\times e^{\frac{\pi i}{2} K} 2^{\frac{K}{2}} + (0.039-0.008i)\times e^{-\frac{\pi i}{2} K} 2^{\frac{K}{2}} + \cdots
	\end{split}
\end{align}
The first line should be combined to give the leading order of large $K$ asymptotics. 
As expected, the second-order degenerate poles ($r_2=r_4$) generate polynomial enhancements (linear in $K$) accompanying the exponential growth, distinct from standard perturbative $1/K$ corrections.
These are precisely the logarithmic corrections to the entropy.

We now turn to the dominant partition configurations underlying this model. 
By expanding the generating function, we find:
\begin{equation}
\mathcal{Z}(x) = \prod_{n=1}^\infty \sum_{p_n=0}^\infty n^{p_n} x^{n p_n} = \sum_{K=0}^\infty \sum_{p\vdash K} \prod_{i=1}^K i^{p_i} x^K = 
\sum_{K=0}^\infty \sum_{p\vdash K} N(p) x^K \,.
\end{equation}
In this scenario, the dominant configuration cannot belong to $\mathcal{P}_1$.
The corresponding weight factor of 1-cycle configurations (i.e., $p_1 = K-m$) is $1^{p_1} = 1$, providing no exponential enhancement.
Let us introduce $\tilde{\mathcal{P}}_l(K)$ to denote the set of partitions that include at least one cycle of length $l$. 
Given the definition of $N(p)$ in \eqref{eq:def-Np}, we assert that the dominant configurations are governed by $\tilde{\mathcal{P}}_3(K)$.

To see this, recall that the cycle multiplicities $p_i$ are subject to the constraint $\sum_{i} i p_i = K$.
Assuming $\tilde{\mathcal{P}}_l$ represents the dominant configuration class, the norm scales as:
\begin{equation}
	N(p) = l^{p_l} \prod_{i \neq l} i^{p_i} = l^{\frac{K}{l}-m} \prod_{i \neq l} i^{p_i} \,.
\end{equation}
In the large $K$ limit, maximizing this exponential growth is equivalent to maximizing the base  $l^{\frac{1}{l}}$.
For integer $l \ge 1$, this function reaches its global maximum at $l=3$.
Without loss of generality, let us focus on the case where 
\begin{equation}
	\tilde{\mathcal{P}}_3 = [3^{L- m_3}, r], \qquad r\vdash 3m_3 \,. 
\end{equation}
The corresponding contribution evaluates to:
\begin{align}
	\begin{split}
\tilde{Z}_3(K) &= \sum_{m_3=0}^{L-1} 3^{L-m_3} \sum_{r\vdash 3m_3, 3\notin r}\prod_{i\neq 3} i^{p_i} \\
& \sim 3^L\sum_{m_3=0}^{\infty} 3^{-m_3} \Coeff\left(x^{3m_3}, \prod_{n\neq 3} \frac{1}{1-nx^n} \right) \,.
 	\end{split}
\end{align}
This combinatorial derivation precisely recovers the leading asymptotic term for $K=3L$ presented in \eqref{eq:ZKwn=n}, confirming the partition set $\tilde{\mathcal{P}}_3$ contributes the leading order of the degeneracy.

More generally, for models governed by an arbitrary non-constant weight sequence $w_n$, the critical cycle length $\ell_\star$ characterizing the dominant partition configuration $\tilde{\mathcal{P}}_{\ell_\star}$ is strictly determined by the integer that maximizes the value of $w_n^{1/n}$. Equivalently, from a complex-analytic perspective, this corresponds to the index that minimizes the pole radius $r_n = w_n^{-1/n}$. This underlying combinatorial principle provides a robust and universal framework for systematically extracting the all-orders asymptotics across a wide variety of weighted partition models.

\section{Asymptotic probability distributions of trace structures  } 
\label{sec:trace-picture}

We showed in section \ref{sec:MTconeq} that the large $ K = (m+n) $ asymptotics of the 
counting function $ \cZ(m,n ) $  of 2-matrix invariants, and of the unrefined version $ \cZ_d (x )$, 
is organised by the cycle structures of the 
equivalence permutations $ \gamma$ in \eqref{permequivs}. These permutations relate different permutations 
$ \sigma \in S_{ m+n}$ which determine the contraction of  the upper indices of $ X^{ \otimes m } \otimes Y^{ \otimes  n }$ as in 
\eqref{Osigma}.  The cycle structure of $ \sigma $ is directly related to the trace structure of the invariant. 
The natural question then arises:  what can be said about the asymptotic distribution of the trace structures (the cycle structures of $\sigma$) themselves, in the context of the all-orders asymptotic expansion we have developed ?

To answer this, we return to the Burnside formula \eqref{Zmnfxdpt}, rewritten here for convenience as
\bea
Z ( m , n ) = { 1 \over m! n! } \sum_{ \gamma \in S_m \times S_n } \sum_{ \sigma \in S_K } \delta ( \sigma \gamma \sigma^{-1} \gamma^{-1} )  \,.
\eea
The cycle structure of $ \gamma$ is specified by a partition $p \vdash m$ and a partition $q \vdash n$. 
By reorganizing the sum over conjugacy classes, the counting function can be expressed as:
\begin{equation}
Z ( m , n ) =  \sum_{ p \vdash m }  \sum_{ q \vdash n } { 1 \over (\Aut_m p) (\Aut_{ n} q )}  \Big| G ( p,q)  \Big|  \,,
\end{equation}
where the stabilizer subgroup $ G ( p , q ) $ consists of all permutations $\sigma \in S_K$ satisfying $\sigma \gamma \sigma^{-1} = \gamma$.
This stabilizer subgroup is isomorphic to a product of wreath products over the cycle lengths $i$ (see  \cite[Chapter 1]{MacdonaldSymmetricFunctions}, \cite[Chapter 1]{JamesKerberSymmetricGroup} for further information on these groups)
\begin{equation}
G ( p , q ) \equiv \prod_{ i=1}^K S_{ p_i +q_i  }  [ \mathbb{Z}_i ] \subset S_{ K } \,.
\end{equation}
Here, $p_i+q_i$ is the total number of cycles of length $i$ in $\gamma$, $\mathbb{Z}_i $ acts within individual cycles, and $S_{p_i+q_i}$ permutes cycles of the same length. 
The order of this group is thus exactly $|G(p,q)| = \prod_i i^{ p_i + q_i } ( p_i + q_i ) ! $.

The  cycle structures of the permutations in this wreath product are  the cycle structures of $ \sigma $,  which determined the trace structures. This is encoded in a known function which we will call the cycle-counting-function $ Z_{\rm  cyc. count.  } ( G ( p , q ) ; P )$  for the subgroup $ G ( p , q ) $, closely related to the cycle index  of the $  G ( p , q ) $ (See \cite{Cameron1994,GouldenJackson1983} for further information on wreath products and their cycle indices.). 
The  $P$ is a partition of $K$ which specifies a conjugacy class in $S_{ K} $. We have 
\bea\label{cyccountfun} 
&& Z_{\rm  cyc. count.  } ( G ( p , q ) ; P ) = \hbox { Number of permutations with cycle structure $P \vdash K $ } \cr 
&& \hbox{ in the subgroup  } G ( p , q ) \hbox{ of } S_K  \,.
\eea
There are nice generating functions for these cycle counting functions of wreath products of the form at hand. 
We can now write the counting function of multi-traces as a sum over $ p \vdash m , q \vdash n , P \vdash (m+n)$
\bea\label{PFcyccountfun} 
\boxed{ 
Z ( m , n )  = \sum_{ p \vdash m } \sum_{ q \vdash n } { 1 \over \Aut_m p \Aut_{ K -m} q }  \sum_{ P \vdash K } Z_{\rm  cyc. count.  } ( G ( p , q ) ; P )
}  \,.
\eea
For fixed $p,q$, we have a distribution over trace-structures (cycle-structures of $ \sigma $) given by 
$
Z_{\rm  cyc. count.  } ( G ( p , q ) ; P ) \,.
$
When we specify some small cycles for $(p,q)$ contributing to the leading order of the asymptotics, for each specified cycle structure, we have a distribution of multi-trace structures given by this cycle count function. 

The key lesson is that the  leading asymptotics which we have computed using the circles of poles method  localises  in cycle structures of $ \gamma$  - specifically small cycles - but is  a probability distribution   over cycles structures in $ \sigma $ ( i.e. matrix trace structures) specified by the cycle count function in \eqref{cyccountfun}. We have 
\bea 
\sum_{ P \vdash K } Z_{\rm  cyc. count.  } ( G ( p , q ) ; P )  = |G ( p, q ) | 
\eea 
which can be equivalently recast to a normalized probability distribution:
\bea 
\sum_{ P \vdash K }  { Z_{\rm  cyc. count.  } ( G ( p , q ) ; P )  \over |G ( p, q ) |  }  =1  \,.
\eea
The summands are positive and add up to one, so can be interpreted as probability distributions over cycle structures of $ \sigma $, i.e. over $P \vdash K$, i.e. over multi-trace structures, for each choice of $(p,q)$. 

In terms of the probabilities we can write \eqref{PFcyccountfun} as 
\bea\label{PFTraceStructProb} 
Z ( m , n )  && = \sum_{ p \vdash m } \sum_{ q \vdash n } {  | G ( p , q ) | \over (\Aut_m p)( \Aut_{ K -m} q) }  \sum_{ P \vdash K } { Z_{\rm  cyc. count.  } ( G ( p , q ) ; P ) \over | G ( p , q ) | }  \cr 
&& = \sum_{ p \vdash m } \sum_{ q \vdash n } \prod_{ i } {  ( p_i + q_i ) ! \over p_i! q_i! } \sum_{ P \vdash K } { Z_{\rm  cyc. count.  } ( G ( p , q ) ; P ) \over | G ( p , q ) | }  \,.
\eea 
In this equation, the sum over partitions $P$ which keep track of the multi-trace structure is weighted by probabilities
that depend on the partitions $ p , q $ of $m,n$ respectively, and increasing cycle lengths on $p,q$ control the orders of the large $K$ asymptotics.

The equations \eqref{PFcyccountfun} and \eqref{PFTraceStructProb} give an expression for $Z ( m, n ) $ with a sum over the cycle structures $P$ of the index contraction permutations $ \sigma $, which keep track of the trace structures of the matrix invariants, along with a sum over the cycle structures $ p \circ q $ of the permutations $ \gamma$, which are the equivalence permutations. Well-known formulae in the  literature connect with the sum over $ P$, without  the additional refinement of summation over $ p,q$ which keep track of orders of the asymptotic expansion as shown in section \ref{sec:MTconeq}. The counting organized by the trace structure, in the $d$-matrix case,  $\mathcal{Z}_d(x)$ can be written as 
\begin{align}\label{eq:Zdx-tracetocycle}
\begin{split}
	\mathcal{Z}_d(x) &= \prod_{i=1}^\infty \frac{1}{1-dx^i} = \prod_{i=1}^{\infty} \frac{1}{(1-x^i)^{a_i}},\qquad a_i= 	\frac{1}{i} \sum_{q|i} \phi(q) d^{\frac{i}{q}} \,.
	\end{split}
\end{align}
where $\phi(q)$ is the Euler totient function. 
The proof of \eqref{eq:Zdx-tracetocycle} is straightforward via standard combinatorial enumeration \cite{FlajoletSedgewick:2009}.
The exponents $a_i$ exactly count the degeneracies of single-trace operators (or single-letter indices), which generate the full multi-trace partition function $\mathcal{Z}_d(x)$ via the Plethystic exponential \cite{Feng:2007ur}. 
These expressions are reviewed in section 5 of \cite{LBSR1} and related to the mathematical literature on the combinatorics of Lyndon words. 
The identity  \eqref{eq:Zdx-tracetocycle} provides a combinatorial formula as a sum over $P \vdash K$:
\begin{equation}\label{eq:Z2K-trace-coeff-Com}
	Z_d(K) = \sum_{P\vdash K} \prod_{j=1}^K C_{P_j+a_j-1}^{P_j}, \qquad P=[1^{P_1},2^{P_2},\cdots K^{P_K}] \,.
\end{equation}
which is related to \eqref{PFcyccountfun} and \eqref{PFTraceStructProb} after sums over $p,q$ in these equations. 

A more systematic study of the asymptotic distributions of all the trace structures using \eqref{PFTraceStructProb} is an important 
direction for the future. This will connect with studies of the large $K$ counting of single traces (e.g. in the context of small black holes in AdS/CFT \cite{Berenstein:2018lrm}) and take into account the multiplicity of energy eigenstates arising within single traces, along the plethystic exponentiation of the single traces.

\section{Summary and Outlook }\label{sec:discussion}

\subsection*{Summary}

In this work we have developed an analytic framework for the degeneracies $Z_d(K)$
associated with $d$-matrix partition functions, which arise in the counting of
gauge-invariant operators in supersymmetric sectors of large-$N$ gauge theories.
By analysing the singularity structure of the generating function
\begin{equation}
	\mathcal{Z}_d(x)=\prod_{i=1}^{\infty}\frac{1}{1-dx^i},
\end{equation}
we showed that the coefficients admit an all-orders asymptotic expansion obtained
by systematically subtracting pole contributions located on concentric circles in
the complex $x$-plane. The dominant pole at $x=1/d$ encodes the familiar Hagedorn
growth of states, while subleading poles generate exponentially suppressed
corrections refining the degeneracy formula. These contributions are naturally
organised by an integer $n$ labelling the pole layers, which admits a combinatorial
interpretation in terms of permutations whose minimal cycle length is $n$, leading
to a picture of \emph{small-cycle dominance} in the asymptotics of multi-trace
operators. An interesting physical aspect is that the subleading contributions
beyond the Hagedorn singularity encode information about the analytic continuation
of the partition function and ultimately reflect the low-temperature expansion
around $x=0$. In particular, for $d\ge13$ the coefficients of the pole expansion
become exponentially suppressed, rendering the series absolutely
convergent; in this regime the pole expansion provides a convergent
 representation of the partition function for $ |x| < d^{-1} $, which can be analytically continued inside the unit disc. 
 This can be viewed as a generalisation of the Mittag-Leffler expansion, where the usual  condition of boundedness on the complex plane is relaxed.  Together these results provide a unified analytic and
combinatorial description of the large-energy behaviour of multi-matrix spectra
and clarify how the singularity structure of the partition function controls the $d^K$ 
Hagedorn growth and its sub-exponential corrections.

From a broader mathematical perspective, our results highlight a class of generating
functions whose analytic structure is characterised by infinitely many simple poles
arranged on concentric circles and accumulating at the boundary $|x|=1$. While
meromorphic generating functions have long played an important role in analytic
combinatorics and partition theory, the $d$-matrix partition functions studied here
provide a particularly transparent setting where the poles can be organised
systematically and their contributions summed to obtain an all-orders asymptotic
expansion for the coefficients. The resulting hierarchy of exponential terms admits
a natural combinatorial interpretation in terms of permutations organised by their
minimal cycle length. As discussed in Section~\ref{sec:CircSings}, the same circle-of-poles framework
extends to generating functions associated with weighted partitions, indicating that
this analytic structure is not specific to the matrix model context. In this way the
present work builds on the extensive literature on meromorphic generating functions
while identifying a subclass—characterised by poles accumulating at the unit
boundary—for which the pole expansion provides an effective tool for extracting
detailed asymptotic information, resulting in some cases in convergent expansions. 

\subsection*{Discussion}

Beyond the explicit asymptotic formulas, our analysis reveals several implications for the underlying mathematical physics.
\begin{enumerate}
\item \textbf{UV/IR reconstruction}: The analytical distinction between the regimes $d \le 12$ and $d \ge 13$ admits a physical interpretation in terms of a UV reconstruction. 
For dimensions $2 \le d \le 12$, the large $K$ asymptotic expansion of the $d$-boson matrix partition function is formally divergent, whereas it transitions to absolute convergence for $d \ge 13$.  
Technically, the leading asymptotic coefficient is systematically extracted by evaluating the large-$K$ limit of $Z_d(K)/d^K$.
By subsequently subtracting this leading $d^K$ contribution, dividing the remainder by $d^{K/2}$, and taking successive large-$K$ limits, one iteratively determines the subleading coefficients. 
This sequential procedure constitutes a high-energy, or ultraviolet (UV), reconstruction of the exact degeneracies $Z_d(K)$, which are the Taylor expansion coefficients in the small $x$ low energy or infra-red (IR) region. 
The absolute convergence we have proved  and the numerical evidence  for $d \ge 13$ suggests that this UV reconstruction is exact and complete. 
Conversely, the divergence observed for $2 \le d \le 12$ indicates that recovering the exact finite-$K$ physics in lower dimensions cannot be achieved exclusively through the UV all-orders asymptotic data; rather, it requires additional 
 input, which may be associated geometrically with the natural boundary at $ |x| =1$. 

\item \textbf{Comparison to resurgence}: In the $d \le 12$ regime, the divergent remainder isolated via optimal truncation encodes further contributions to be determined, a property formally aligned with the resurgence framework \cite{Dorigoni:2014hea,Dunne:2015eaa,Aniceto:2018bis}.
Because the discrete pole expansion strictly inside the open unit disk already captures  exponential scales (of order $d^{\frac{K}{n}}$), similar to suppression of the form Farey tail or roots of unity saddles \cite{Dijkgraaf:2000fq,Manschot:2007ha}, the remaining error evaluated at the optimal truncation order represents sub-exponential  corrections (as also evidenced  by numerical tests in $d=2$). 
Geometrically, these residual terms are associated with the accumulation of singularities on the unit circle. 
This mathematical structure parallels the emergence of modular properties in the Jacobi theta function, where the exponentially suppressed terms required to complete the modular transformation originate precisely from continuous boundary integrals (See chapter 2 of \cite{elizalde1994zeta}). 
Similarly, resolving the sub-exponential pieces in the $(d\le 12)$ $d$-matrix partition function is likely to require continuous integrals evaluated along the natural boundary at $|x|=1$.

\item \textbf{Complementary region in complex plane}: The general $q$-Pochhammer symbol $(x;q)_{\infty}$, arising in the $d$-matrix partition function (equation \eqref{eq:partition-unrefine}), also arises in the holomorphic blocks and partition function for three-dimensional superconformal indices \cite{Beem:2012mb,Yoshida:2014ssa}, where its asymptotic behaviour is essential for extracting the corresponding density of states. 
Recent work \cite{ArabiArdehali:2026ddr} has analysed the asymptotics of this function strictly within the $|x|<1$ regime. 
A key geometric distinction in this context is that all relevant singularities of the $q$-Pochhammer symbol lie either on or outside the unit circle. Consequently, the complex-analytic pole-extraction methodology developed in our work is not directly applicable for capturing the leading-order asymptotics in this regime; instead, the appropriate analytic framework requires the study of modular transformations. In this sense, our analysis and the approach of \cite{ArabiArdehali:2026ddr} provide complementary perspectives on the full asymptotic structure of the $q$-Pochhammer symbol. 
\end{enumerate}

\subsection*{Future directions }
A natural progression in understanding the operator counting algorithm for the $PSU(1,2|3)$ subsector \cite{Baiguera:2022pll}— which potentially includes states dual to $\frac{1}{16}$-BPS AdS black hole \cite{Kunduri:2006ek} — is to generalize our complex-analytic methodology to larger closed subsectors of $\mathcal{N}=4$  SYM \cite{Harmark:2007px,Beisert:2004ry}. 
The simplest such extension is the $SU(1|2)$ subsector, which incorporates the two scalars analysed in this work alongside an additional chiral fermion.
This sector holds particular physical relevance due to its mapping to the $t-J$ model, a foundational framework in the theory of high-$T_c$ superconductivity \cite{Beisert:2005fw}.
The corresponding unrefined partition function is given by:
\begin{equation}\label{eq:SU23-PARTITIONFUNCTION}
	\mathcal{Z}(x)=		\prod_{n=1}^\infty \frac{1}{(1-x^{n})(1+x^{n}-x^{2n})} \frac{1}{(1+x^{n-\frac{1}{2}})(1-x^{n-\frac{1}{2}}-x^{2n-1})} \,.
\end{equation}
This partition function exhibits a hybrid analytic structure. 
It contains non-degenerate poles, with the dominant singularity located at $x=\frac{3-\sqrt{5}}{2}$.
This pole yields a Hagedorn temperature consistent with \cite{Harmark:2006ta,Harmark:2007px} and dictates an exponential leading-order degeneracy growth of:
\begin{equation}
	Z(K) \sim \left(\frac{3+\sqrt{5}}{2} \right)^K \sim 2.618^K \,.
\end{equation}
Simultaneously, the partition function possesses an infinite set of singularities akin to those of the Dedekind eta function, which independently drive a sub-exponential Hardy-Ramanujan-type growth. 
Besides it also contains singularities outside the natural boundary at $x= \frac{1+\sqrt{5}}{2}$ \footnote{Models contain both singularities interior and exterior to the natural boundary were also recently studied in \cite{Lee:2025veh} within the context of AdS$_3\times S^3\times \mathcal{M}_4$.}. 
The intricate interplay between these exponential and sub-exponential mechanisms significantly complicates the exact extraction of subleading corrections.
It is anticipated that even richer analytic structures govern the partition functions of non-compact subsectors, such as $PSU(1,1|2)$ \cite{Baiguera:2021hky} and $SU(1,2|2)$ \cite{Baiguera:2020mgk} necessitating further mathematical development.
A further interesting generalization is to turn on the 't Hooft coupling \cite{Spradlin:2004pp,Suzuki:2017ipd} in the Spin Matrix theory limit where the higher order Feynman diagrams are suppressed.

Beyond $\mathcal{N}=4$ SYM subsectors, similar analytic structures naturally emerge in the counting of local operators for general quiver gauge theories \cite{Pasukonis:2013ts} (See also \cite{McGrane:2015jza}). 
Our complex-analytic methodology provides a natural framework to decode the asymptotic degeneracies of these generalized configurations.
Examples include the Klebanov-Witten conifold and the $\mathbb{C}^3/\mathbb{Z}_2$ orbifold, whose unrefined large-$N$ partition functions take the respective forms:
\begin{align}
	\begin{split}
\text{Conifold}: & \qquad  \mathcal{Z}(x) = \prod_{n=1}^\infty \frac{1}{1-4 x^{2n}} \\
\mathbb{C}^3/\mathbb{Z}_2: & \qquad \mathcal{Z}(x) = \prod_{n=1}^\infty \frac{1}{(1-3x^n) (1+x^n)} \,.
	\end{split}
\end{align}
The conifold partition function structurally is identical to a $4$-matrix theory evaluated at $x^2$, making our asymptotic framework directly applicable. 
Conversely, the partition function of the $\mathbb{C}^3/\mathbb{Z}_2$ orbifold mirrors a hybrid analytic structure of the $SU(1|2)$ model \eqref{eq:SU23-PARTITIONFUNCTION}, combining the isolated poles of a $3$-matrix theory with an infinite family of dense boundary singularities from the $(1+x^n)^{-1}$ factor. 
We leave the precise asymptotic resolution of these hybrid boundary structures to future work.

The identification of the critical dimension $d\ge 13$ at which the asymptotic expansion of the bosonic matrix theory degeneracy transitions to absolute convergence is reminiscent of the known results in gravitational theory \cite{Sorkin:2004qq,Suzuki:2015axa,Kol:2006vu}. In the context of higher-dimensional general relativity, this specific dimension defines the boundary governing the order of the phase transition between uniform and non-uniform black strings, a phenomenon fundamentally driven by the Gregory-Laflamme instability \cite{GregoryLaflamme,HorowitzMaeda}. 
Specifically, for spacetime dimensions $D \le 13$, this instability is driven by a first-order phase transition, whereas for $D \ge 14$, the transition becomes continuous (higher-order) \cite{Harmark:2005pp,Kol:2004ww}.
It is plausible that the transition we found in the bosonic $d$-matrix theory could be a  feature shared by more  general matrix theories within an appropriate universality class. Such theories  could potentially include ones having  a bulk gravitational dual with one dimension higher. In such a hypothetical scenario, the $D\ge 13$ region for convergence which we have found would map to  a bulk transition at $ D \ge 14$. it would remain to understand whether there is a physical connection between the convergence property of the partition function with the continuity in the phase diagram of the black string. Nevertheless, the emergence of an identical critical dimension in both the exact combinatorial asymptotics of matrix theory and the thermodynamic instability of gravity points to a potential universality inherent to s large-dimension expansions in gravity  which have been studied in \cite{Emparan:2013moa,Emparan:2020inr}. 
We leave the investigation of  more precise links between these phenomena for the future. 

Our results on the analytic  properties of the $d$-matrix partition function, interpreted as a UV reconstruction as discussed above,   suggest that
for $d \le 12$ the reconstruction of finite-$K$ physics requires an interplay
between ultraviolet and infrared data, whereas for $d \ge 13$ the ultraviolet
asymptotics alone suffice. 
Multi-matrix models have shown an unreasonable effectiveness as models of fundamental physics. 
Prominent examples include the BFSS matrix model for flat-space M-theory in eleven dimensions \cite{BFSS} and the IKKT matrix model for non-perturbative type IIB string theory \cite{Ishibashi:1996xs}. 
Motivated by this, one may speculate that the critical dimension emerging in the bosonic $d$-matrix model reflects a deeper structural feature: the necessity of UV–IR complementarity in reconstructing finite-energy physics. It would therefore be interesting to explore further evidence for such a broader picture. 

Dimensional thresholds have often played an important role in
fundamental physics. A well–known example is the Brandenberger–Vafa
mechanism \cite{BrandenbergerVafa}, which attempts to explain the emergence of four large spacetime
dimensions starting from string theory. From the perspective of the present
matrix model analysis, one may ask a different question: assuming that an
underlying matrix description of spacetime plays a fundamental role, why do the
dimensions in which string and M- and F-theory structures appear ($D=10,11,12$)
lie below the critical dimension identified here? 
Also, are there dimension-dependent physical observables that experience a phase transition near the critical dimensions?

These observations suggest that UV–IR complementarity may play a broader
role in the reconstruction of finite-energy physics in string theory. Matrix models provide
a particularly tractable setting in which this phenomenon can be studied
quantitatively, but related mechanisms could conceivably appear in other
contexts where different  degrees of freedom encode spacetime physics, ranging from
microscopic descriptions such as brane dynamics to effective frameworks
including supergravity and effective field theories.

\section*{Acknowledgements}

We thank  Joseph Ben Geloun, Jie Gu, Robert de Mello Koch, Sam van Leuven and Denjoe O' Connor for useful discussions. 
Both authors thank each other's institutes hospitality where this project is in progress.
Y.L. is supported by a Project Funded by the Priority Academic Program Development of Jiangsu Higher Education Institutions (PAPD) and by National Natural Science Foundation of China (NSFC) No.12305081.
S.R. is supported by the
Science and Technology Facilities Council (STFC) Consolidated Grant ST/X00063X/1
“Amplitudes, strings and duality”. 
We are both supported by the ongoing
Royal Society International Exchanges grant IEC\textbackslash NSFC\textbackslash 242376 held jointly with the National Natural Science Foundation of China No.W2421035. AI tools, notably ChatGPT and Gemini, were used to assist in refining the clarity of exposition and the presentation of mathematical arguments; all scientific results and interpretations are those of the authors.

\appendix

\section{Numerical results of $Z_2(K)$ and its asymptotic features}
\label{appen:numericalZ2}

The asymptotic nature of the series $Z_2(K)$ is further characterized by the "optimal truncation" scheme. 
We partition the non-negative integers into disjoint intervals $\{\mathcal{C}_M  \}_{M\ge1}$ defined such that for a given $K \in \mathcal{C}_M$, the truncation at order $M$ minimizes the error:
\begin{equation}
	K\in \mathcal{C}_M \quad \Rightarrow \quad	|	Z_2(K;M)- Z_2(K)| = \min_{M'\ge 1}	|Z_2(K;M')- Z_2(K)| \,.
\end{equation}
Let $K_{\pm}(M)$ denote the upper and lower bounds of each optimal interval $\mathcal{C}_M$. Table \ref{Table:optimaltruncation} lists the truncation intervals and corresponding errors for $M=1$ to $10$. It is observed that while the absolute error increases with $K$ as a sign of divergent series, the relative error $|Z_2(K;M) - Z_2(K)| / Z_2(K)$ decreases substantially, confirming the asymptotic convergence of the expansion.
\begin{table}[]
	\centering
	\begin{tabular}{|c|c|c|c|c|}
		\hline
		Truncation & 
		\begin{tabular}{@{}c@{}}Optimal $\mathcal{C}_M$ \\
			$K_- \le K \le K_+$
		\end{tabular} & 
		Order of $Z_2(K_+)$ & $Z_2(K_+;M) - Z_2(K_+)$&
		Ratio $\frac{\text{Error}}{Z_2(K_+)}$ \\
		\hline
		$M=1$ & $K\le15$ & $10^5$ & 4000 & $10^{-2}$ \\
		\hline
		$M=2$ & $16\le K\le 38$ & $10^{12}$ &$10^6$ & $10^{-6}$ \\
		\hline
		$M=3$ & $39\le K\le 71$ & $10^{21}$ & $10^9$ & $10^{-12}$ \\
		\hline
		$M=4$ & $72 \le K\le 115$ & $10^{35}$ & $10^{12}$& $10^{-23}$ \\
		\hline
		$M=5$ & $116\le K\le 169$ & $10^{51}$ & $10^{14}$ & $10^{-37}$ \\
		\hline
		$M=6$ & $170\le K\le 233$ & $10^{70}$ & $10^{17}$ & $10^{-53}$ \\
		\hline
		$M=7$ & $234 \le K\le 308$ & $10^{93}$ & $10^{20}$ & $10^{-73}$ \\
		\hline
		$M=8$ & $309 \le K\le 391$ & $10^{118}$ & $10^{23}$& $10^{-95}$ \\
		\hline
		$M=9$ & $392\le K\le 486$ & $10^{146}$ &$10^{26}$ & $10^{-120}$ \\
		\hline
		$M=10$ & $487\le K\le 592$ & $10^{178}$ & $10^{29}$& $10^{-149}$ \\
		\hline
	\end{tabular}
	\caption{
		Truncation scheme and error analysis for $Z_2(K)$ at $K_+(M)$. For each optimal truncation order $M$ and its interval $\mathcal{C}_M$, the table lists the order of magnitude of $Z_2(K_+)$ and the corresponding truncation error. 
		In the last column, we show while the error $Z_2(K_+;M) - Z_2(K_+)$ increases with $K$, its ratio to $|Z_2(K_+)|$ keeps decreasing dramatically.
		We can also notice the error term roughly grows as $10^{3M}$.
	}
	\label{Table:optimaltruncation}
\end{table}

The Table \ref{Table:optimaltruncation} also confirms the asymptotic expansion \eqref{eq:degeneracyZ2} and \eqref{eq:numer-Z2K-suble} matching with the numerical tests.
Another way to see the asymptotic nature of the series is by its optimal truncation. 
Take the major contribution (i.e. positive real root) in the asymptotic expansion:
\begin{equation}\label{eq:summationseries}
	Z_2(K) \sim  \sum_{n}^\infty \sqrt{\frac{\ln 2}{2\pi}} (-1)^{n-1} \frac{1}{n^{\frac{3}{2}}} \exp \left( \frac{\pi^2}{4 \ln 2} n  \right) 2^{\frac{K}{n}}\,, \qquad n\to \infty \,.
\end{equation}
For such series, the error is approximated by the next term. At the optimal truncation, this error is minimised which  amounts to extremizing the summand.  
Applying the saddle-point approximation to the exponent in \eqref{eq:summationseries}, the extremization condition
\begin{equation}
	\frac{d}{dn} \left[ \frac{\pi^2}{4\ln 2} n + \frac{K}{n} \ln 2\right] =0 \,
\end{equation}
yields the optimal truncation order $M$ for a given $K$:
\begin{equation}\label{eq:saddlepoint-Kn}
	K= \frac{\pi^2}{4\ln^2 2} M^2 \approx 5.1356 M^2 \,.
\end{equation}
The close agreement with the numerical fit $K_+(M)$ listed in Table \ref{Table:optimaltruncation} are well-fitted by the quadratic function
\begin{equation}\label{eq:growthofKM}
	K_+(M) \approx 5.1174 M^2+ 7.7811 M+1.9833 \,,
\end{equation}
(from Table~\ref{Table:optimaltruncation}) validates that the large-$n$ behaviour of the $j=0$ poles accurately captures the numerical convergence properties of the model.

This analytic structure however, also has indications to the error term between the exact $Z_2(K)$ and the series approximation $Z_2(K;M)$ in \eqref{eq:Z2-truncation}.
From  the Table~\ref{Table:optimaltruncation}, the errors are asymptotically $10^{3M}$. Combined with the optimal truncation region \eqref{eq:growthofKM}, we can conjecture the errors for near the optimal truncation grows as $e^{\lambda \sqrt{K}}$ for some constant $\lambda$.
This is interesting as this subexponential growth usually originates from the generating partition function of Euler partition number up to some power $c$: $\mathcal{Z}_1(x)^c$.
In another word, they are the contributions of singularities precisely from the boundary of unit disk.

Indeed, our formula \eqref{eq:degeneracyZ2} shares a foundational philosophy with the Rademacher expansion, despite the differences in their analytic implementation. Both approaches rely on a two-step reconstruction of the degeneracy $Z(K)$ from the singular data of the partition function:
\begin{itemize}
	\item \textbf{Local Characterization}: In the Rademacher case, one exploits modular symmetry to determine the behaviour near essential singularities. In our treatment of $\mathcal{Z}_2(x)$, we characterize the local behaviour by identifying the residues of the discrete simple poles located at $x_{n;j}$.
	\item \textbf{Global Summation}: Instead of performing an inverse Laplace transform over Ford circles, we employ a partial fraction expansion (formula \eqref{eq:pole-expansion}) to sum the contributions of these poles. By expanding each partial fraction as a geometric series, the coefficients are extracted directly, bypassing the need for complex contour deformations like Ford circles.
\end{itemize}
The expansion \eqref{eq:pole-expansion} thus functions as a Farey-tail-like expansion, parametrized by rational numbers $j/n$. 
It suggests that the $2$-matrix model, much like 2D CFTs, admits a representation where the total degeneracy is a coherent sum over "saddles" (poles) indexed by rational fractions $j/n$, effectively bridging the discrete analytic structure of matrix models with the universal features of holographic systems \cite{Alday:2019vdr}.

\section{Detailed computations of the subleading asymptotics by combinatoric method}
\label{append:higher-order-combinartorics}

In this appendix, we will provide details in computing the subleading remainder terms \eqref{eq:def-cR10}. 
The sum over $n_1$ acts as roots of unity filter of the generating function $(2x^3;x)_{\infty}^{-1}$. 
When  $ ( K + m_1') $ is even ($\bar{E}=0$ in \eqref{eq:def-subleadingZ2-config}), it extracts the even part of $(2x^3;x)_{\infty}^{-1}$; when odd ($\bar{E}=1$), it extracts the odd part scaled by a factor of $\sqrt{2}$. 
Synthesizing these parity projections yields the exact closed-form expression:
\bea
&&  \cR_{ 1,0} ( K )  = \sum_{ m_1' =0 }^{ \infty } 
2^{\lfloor  { K+ m_1' \over 2 }    \rfloor  - m_1'}   
\biggl [   { 1 \over 2 }   \left (\frac{1}{(2x^3;x)_{\infty}} \Big|_{x=\frac{1}{\sqrt{2}}} 
+\frac{1}{(2x^3;x)_{\infty}} \Big|_{x=-\frac{1}{\sqrt{2}}}   \right )  \left ( 1 - \bar E ( K + m_1' ) \right ) 
\cr    
&& + { \sqrt { 2 }  \over 2   }   \left (\frac{1}{(2x^3;x)_{\infty}} \Big|_{x=\frac{1}{\sqrt{2}}} 
-\frac{1}{(2x^3;x)_{\infty}} \Big|_{x=-\frac{1}{\sqrt{2}}}   \right )   \bar E ( K + m_1' )
\biggr ]  \cr 
&& = \sum_{ m_1' =0 }^{ \infty } 
2^{\lfloor  { K+ m_1' \over 2 }    \rfloor  - m_1'}   
\biggl [   { 1 \over 2 }  
\left ({ 1 \over  ( { 1 \over \sqrt  2 } ; { 1 \over \sqrt 2 } )_{\infty} }  +  { 1 \over  ( -  {1 \over  \sqrt 2 } ; -  { 1 \over \sqrt 2 }   )_{\infty} }    \right )
\left ( 1 - \bar E ( K + m_1' ) \right ) 
\cr    
&& + { \sqrt {2}  \over 2  }  
\left ({ 1 \over  ( { 1 \over \sqrt  2 } ; { 1 \over \sqrt 2 } )_{\infty} }   -  { 1 \over  ( -  {1 \over  \sqrt 2 } ; -  { 1 \over \sqrt 2 }   )_{\infty} }    \right )
\bar E ( K + m_1' )
\biggr ] \,.
\eea 

When $K$ is even ($K = 2L$), the exponent simplifies depending on the parity of $m_1'$:
\begin{equation}
	\lfloor \frac{ K+ m_1' }{ 2 } \rfloor - m_1' = L + \lfloor \frac{ m_1' }{ 2 } \rfloor - m_1' = 
	\begin{cases} 
		L - k_1' & \text{if } m_1' = 2k_1' \,, \\
		L - k_1' - 1 & \text{if } m_1' = 2k_1' + 1 \,.
	\end{cases}
\end{equation}
Therefore the remainder term for even $K$ is
\bea\label{simpR10}  
&&  \cR_{ 1,0} ( 2L )  
= \sum_{ k_1' = 0 }^{ \infty } 2^{ L - k_1' } \cdot { 1 \over 2 } \cdot  
\left ({ 1 \over  ( { 1 \over \sqrt  2 } ; { 1 \over \sqrt 2 } )_\infty }  +  { 1 \over  ( -  {1 \over  \sqrt 2 } ; -  { 1 \over \sqrt 2 }   )_\infty }    \right )   \\
&& 
+ \sum_{ k_1' =0 } 2^{ L - k_1' -1  } \cdot { \sqrt {2}  \over 2 } \cdot  \left ({ 1 \over ( { 1 \over \sqrt  2 }; { 1 \over \sqrt 2 } )_\infty }   -  { 1 \over  ( -  {1 \over  \sqrt 2 } ; -  { 1 \over \sqrt 2 })_\infty }    \right )   \cr 
&& = 2^{ L } . { 2 \over 2  }   \left ({ 1 \over  ( { 1 \over \sqrt  2 }; { 1 \over \sqrt 2 } )_\infty }  +  { 1 \over  ( -  {1 \over  \sqrt 2 } ; -  { 1 \over \sqrt 2 } )_\infty }    \right )  
+    { 2^{ L  } \sqrt{2} \over 2  } \left ({ 1 \over  ( { 1 \over \sqrt  2 }; { 1 \over \sqrt 2 })_\infty }   - { 1 \over ( -  {1 \over  \sqrt 2 }; -  { 1 \over \sqrt 2 } )_\infty }    \right )   \nonumber \,.
\eea
This result can be simplified to \eqref{simpR10-even}. The derivation for odd $K$ follows the same steps and yields the analogous formula.

\section{Proof of proposition 5.2} 
\label{appen:proof}

In this section, we will provide a proof of Proposition 5.2 in the section \ref{ssec:all-order-meromo}.

\begin{proof} 
	Applying Lemma 5.2 to the remainder function
	\[
	R_m(z) := F(z) - \sum_{q=1}^m F_q(z),
	\]
	we analyse its singularity structure. By construction, all poles of $F(z)$ with
	$|z| \le r_m$ have been removed, so $R_m(z)$ is analytic in the disk $|z| < r_{m+1}$,
	and its nearest singularities are precisely the poles on the circle $|z| = r_{m+1}$.
	
	We can consider  radii $\tilde r_{m+1}$ such that
	\[
	r_{m+1} < \tilde r_{m+1} \le r_{m+2},
	\]
	so that $R_m(z)$ is meromorphic in $|z| < \tilde r_{m+1}$ with only finitely many poles,
	all located on $|z| = r_{m+1}$. By Lemma 5.2, the coefficients of $R_m(z)$ admit an
	asymptotic expansion determined entirely by these poles.
	
	By definition, $F_{m+1}(z)$ is the sum of the principal parts of $F(z)$ at the poles
	on the circle $|z| = r_{m+1}$. Since these are exactly the poles of $R_m(z)$ in the
	disk $|z| < \tilde r_{m+1}$, we can write
	\[
	R_m(z) = F_{m+1}(z) + H_{m+1}(z),
	\]
	where $H_{m+1}(z)$ is analytic in $|z| < \tilde r_{m+1}$.
	
	Taking coefficients, we have
	\[
	[z^K] R_m(z) = f_{m+1;K} + [z^K] H_{m+1}(z).
	\]
	By Lemma 5.1, the coefficients of $H_{m+1}(z)$ satisfy
	\[
	[z^K] H_{m+1}(z) \simexp  \,  r_{m+2}^{-K},
	\]
	while
	\[
	f_{m+1;K} \simexp  \, r_{m+1}^{-K}.
	\]
	Since $\tilde r_{m+1} > r_{m+1}$, it follows that
	\[
	[z^K] H_{m+1}(z) = o(f_{m+1;K}).
	\]
	Therefore,
	\[
	[z^K] R_m(z) = f_{m+1;K} \bigl(1 + o(1)\bigr),
	\]
	which implies
	\[
	f_K - \sum_{q=1}^m f_{q;K} \sim f_{m+1;K}.
	\]
	This establishes the strong asymptotic expansion and proves Proposition 5.2.
\end{proof}

\bibliographystyle{JHEP}
\bibliography{bib-bh}

\end{document}